\documentclass[a4paper]{article}

\usepackage[headinclude]{typearea}
\usepackage[english]{babel}           
\usepackage[utf8]{inputenc}
\usepackage[T1]{fontenc}
\usepackage{scrtime}
\usepackage{amsmath,amsthm,amsfonts,amssymb}
\usepackage {color}                                  
\usepackage {pifont}                                 
\usepackage {multicol}                               
\usepackage {times}                                  
\usepackage[numbers]{natbib}                                                     
\usepackage [scaled=.90]{helvet}                     
\usepackage {courier}                                
\usepackage {graphicx}                               
\usepackage {array}                                  
\usepackage {longtable}                              
\usepackage {makeidx}                                
\makeindex
\usepackage{fancyvrb}
\usepackage{enumerate} 
\usepackage{ifpdf}
\usepackage{caption}

\usepackage{setspace}

\usepackage{mathptmx}
\usepackage{marvosym}

\theoremstyle{plain}
\newtheorem{theorem}{Theorem}[section]
\newtheorem{cor}[theorem]{Corollary}
\newtheorem{lemma}[theorem]{Lemma}
\newtheorem{proposition}[theorem]{Proposition}

\theoremstyle{definition}

\newtheorem{definition}[theorem]{Definition}
\newtheorem{remark}[theorem]{Remark}
\newtheorem{conjecture}{Conjecture}

\newcommand{\E}{{\mathbb{E}}}

\newcommand{\N}{{\mathbb{N}}}
\renewcommand{\P}{{\mathbb{P}}}

\newcommand{\R}{{\mathbb{R}}}

\newcommand{\trace}{\mathrm{tr}}

\renewcommand{\P}{{\mathbb P}}

\newcommand{\cD}{{\cal D}}
\newcommand{\cE}{{\cal E}}
\newcommand{\cF}{{\cal F}}

\newcommand{\cL}{{\cal L}}

\newcommand{\be}{\begin{equation}}
\newcommand{\ee}{\end{equation}}
\newcommand{\bea}{\begin{eqnarray}}
\newcommand{\eea}{\end{eqnarray}}
\newcommand{\beast}{\begin{eqnarray*}}
\newcommand{\eeast}{\end{eqnarray*}}
\newcommand{\bproof}{\begin{proof}}
\newcommand{\eproof}{\end{proof}}

\title{Bayesian Dividend Optimization and Finite Time Ruin Probabilities}

\author{Gunther Leobacher \and Michaela Sz\"olgyenyi\,\thanks{The author is supported by the Austrian Science Fund (FWF) Project P21943.} \and Stefan Thonhauser\,\thanks{Stefan Thonhauser is supported by Swiss National Science Foundation (SNF) Project 200021-124635/1.}}

\begin{document}

\date{January 2014}

\maketitle


\begin{abstract}
We consider the valuation problem of an (insurance) company under
partial information. Therefore we use the concept of maximizing discounted future dividend payments. The firm value process is described by a diffusion model with
constant and observable volatility and constant but unknown drift parameter.
For transforming the problem to a problem with complete information, we derive
a suitable filter. The optimal value function is characterized as the
unique viscosity solution of the associated Hamilton-Jacobi-Bellman equation.
We state a numerical procedure for approximating both the optimal dividend strategy and the corresponding value function. Furthermore, threshold strategies are discussed in some
detail. Finally, we calculate the probability of ruin in the
uncontrolled and controlled situation.
\end{abstract}

\vspace{5mm}
\noindent {\it Keywords:} dividend maximization, stochastic optimal control, filtering theory, viscosity solutions, finite time ruin probabilities\\
Mathematics Subject Classification (2010): 49L20, 91B30, 93E20


\centerline{\underline{\hspace*{16cm}}}

M. Sz\"olgyenyi \Letter \\
Department of Financial Mathematics, Johannes Kepler University Linz, 4040 Linz, Austria\\
michaela.szoelgyenyi@jku.at\\

G. Leobacher\\
Department of Financial Mathematics, Johannes Kepler University Linz, 4040 Linz, Austria\\

S. Thonhauser\\
Department of Actuarial Science, University of Lausanne, 1015 Lausanne, Switzerland

\centerline{\underline{\hspace*{16cm}}}


\section{Introduction}

In this paper we are going to study the valuation problem of an (insurance) company. We assume that the firm value process is given by a Brownian motion with drift, and  is absorbed when hitting zero.
In contrast to existing results, the drift parameter is modeled as an unobservable Bernoulli-type random variable and the company can only observe the evolution of its firm value.\\
In an insurance context \citet{definetti1957} proposed the expected discounted future dividend payments as a valuation principle for a homogeneous insurance portfolio. However, one can extend this concept to large companies, not necessarily being insurers. The accumulated dividends are described by an absolutely continuous process such that the company is capable of controlling its dividend rate with the aim of maximizing the expected value of accumulated discounted dividend payments.\\
In mathematical terms we face the problem of determining
\begin{align*}
\sup_{u\in A}\E_x\left(\int_0^\tau e^{-\delta t}u_t\,dt\right)\,,
\end{align*}
where the controlled firm value, controlled by some dividend rate $(u_t)_{t\geq0}$, and lifetime are given by
\begin{align*}
X_t&=x+\theta t+\sigma B_t-\int_0^t u_s\,ds\,,\\
\tau&=\inf\{t\geq 0\,\vert\,X_t\leq 0\}\,.
\end{align*}
The drift parameter is a random variable $\theta \in \lbrace \theta_1, \theta_2 \rbrace$ with known distribution.
In this so-called {\em Bayesian framework} the company can only observe the wealth process and thus faces an optimal control problem under partial information.
The a-priori unknown drift parameter expresses the company's uncertainty on the profitability of some business activities or an uncertainty on the general economic environment in addition to the basic risk represented by the Brownian component.\\

In the insurance context the unknown drift parameter can be interpreted as a residual uncertainty when using a diffusion
approximation of a classical risk reserve process instead of the original model.
For diffusion models with observable parameters this problem and several variants of it have been studied intensively, for example by
\citet{shreve1984}, \citet{jeanblanc1995}, \citet{radner1996}, and \citet{asmussen1997}.
For an overview, the interested reader may consult \citet{schmidli2008}, \citet{albrecher20092}, or \citet{avanzi2009}.
Two recent papers, \citet{jiang2012}, and \citet{sotomayor2011} deal with the dividend problem under a changing economic environment,
described by parameters driven by an {\em observable} Markov chain. However these models still assume full information and therefore differ from the model studied here.\\
The dividend maximization problem is also related to a pure optimal consumption problem of an economic agent. The agent is capable of controlling his/her consumption intensity.\\
Papers pointing towards optimal consumption problems arising in mathematical finance are \citet{hubalek2004} and \citet{grandits2007}, maximizing expected accumulated utility of dividends, and expected utility of accumulated dividends, respectively.\\
In corporate finance a similar problem appears, sometimes in combination with an optimal stopping problem, in liquidity risk modeling.
There the firm value process corresponds to the cash reserve process of a
company, the market value of shares of which is given by expected future
dividend payments, for instance see \citet{decamps2007}. In this framework an
uncertain drift parameter is taken up in \citet{decamps2012},
where the solution of a special case of an associated singular control problem is presented.\\
Models with partial information - in particular hidden Markov models - appear
quite frequently in the literature on portfolio optimization problems, e.g., by
\citet{karatzas2001}, \citet{rieder2005}, and \citet{sass2004},
whereas results relating to actuarial mathematics are more scarce. \citet{gerber1977} uses a
Brownian motion with unknown drift for modeling the value of a single
insurance policy, of which it is a-priori not known whether it is a good or bad
risk. He answers the question of when to optimally cancel the policy, i.e.,
when the insurer should decide that the risk actually corresponds to a bad one.
For a diffusion risk reserve process with parameters generated by a hidden
Markov chain, partial differential equations associated to finite time ruin
probabilities are derived
by \citet{elliott2011}.\\

The main contribution of the present paper is the complete analytical characterization of the
solution of the dividend maximization problem under partial information.
Furthermore, based on the analytical findings, we provide a numerical procedure for determining approximations of the optimal value function and dividend strategy. As a complement, following the path described by \citet{elliott2011},
we consider the associated finite time ruin probabilities for the uncontrolled and the controlled situation.\\
\\
The paper is organized as follows. In Section \ref{sec:Pre} we give a
mathematical description of the model and the optimization problem. In Section
\ref{sec:Filtering} we derive, by means of filtering theory, an estimator for the
unknown drift parameter to overcome uncertainty. The applicability of the
dynamic programming approach and the associated Hamilton-Jacobi-Bellman (HJB)
equation for the filtered optimization problem are given in Section
\ref{sec:HJB}. Section \ref{sec:Verif} contains the complete theoretical
characterization of the optimal value function as the unique viscosity solution
of the associated HJB equation. As the proofs are rather technical, they have been moved to the Appendix. In Section \ref{sec:Num} we describe a numerical method for
calculating approximations to the optimal value function and optimal dividend strategy.
Section \ref{sec:Threshold} deals with the special class of threshold strategies, under which dividends are paid only
if the firm value process exceeds a certain threshold level. The finite time ruin probabilities are considered in Section \ref{sec:Ruin}.
Section \ref{sec:Concl} concludes the paper.


\section{Preliminaries}
\label{sec:Pre}

In the whole paper we consider a filtered probability space $(\cE, \cF, \{\cF_t\}_{t \ge 0}, \P)$ carrying all stochastic quantities which will be introduced in the following.\\
As stated in the introduction, we assume that the firm value of a company is given by
\begin{equation}
\label{eq:dynX1}X_t=x+\int_0^t \theta \,ds+\sigma B_t - L_t\,,
\end{equation}
with initial capital $x>0$, where $\theta$ is the constant unobservable drift, $\theta \in \lbrace \theta_1, \theta_2 \rbrace$, $\theta_1 < \theta_2$, $\sigma$
is the constant and known volatility, and $B=(B_t)_{t \ge 0}$ is a standard Brownian motion.
The accumulated dividend process $L=(L_t)_{t\geq 0}$, $L_0=0$,
is assumed to admit a density $u=(u_t)_{t\geq 0}$, which is bounded, i.e., $u_t \in [0,K]$, $K>0$, such that
\begin{align}
\label{eq:dynL1}dL_t=u_t \,dt\,.
\end{align}
Note that $X$ always corresponds to a certain strategy $u$,
but for avoiding an elaborate notation, we will not make that explicit as long as it causes no ambiguities.\\
For modeling a further uncertainty of the company's firm value in addition to the Brownian component we assume that it is not possible
to observe the drift parameter directly, but we assume knowledge of its (initial) distribution $q:=\P(\theta=\theta_1)=1-\P(\theta=\theta_2)$.\\
We denote the uncontrolled firm value process by $Z=(Z_t)_{t\geq 0}$, which is given through
\begin{align}
\label{eq:dynZ1}Z_t&= x + \theta t+\sigma B_t\,,
\end{align}
and obviously $X_t=Z_t-L_t$.
The observation filtration $\{\cF^{Z}_t\}_{t\ge 0}$ is the augmentation of the filtration generated by $Z$.
This means that the company is able to observe the evolution of its uncontrolled firm value and based on that decides on the dividend strategy,
or equivalently observes controlled firm value and accumulated dividends.\\
\\
The value process associated with a dividend strategy $u$ is defined as
\[
J_t^{(u)}:=\E\left(\int_t^{\tau}e^{-\delta(s-t)}u_s \,ds \, \vline \, \cF_t^{Z}\right)\,,
\]
where $\tau:=\inf\{s\ge t \, \vline \, X_s\le 0\}$ is the time of ruin of $X$
for the corresponding dividend strategy $u$. $\tau$ depends on the strategy $u$
via $X$, and again we will not make this explicit if there is no danger of
confusion.\\
The optimal value process of the optimization problem under study is given by
\[
V_t = \sup_{u\in A}\E\left(\int_t^{\tau}e^{-\delta(s-t)}u_s \,ds \, \vline \,\cF_t^{Z}\right)\,,
\]
where $A$ denotes the set of admissible controls, for which we take the set
of all $\{\cF^Z_t\}_{t\ge 0}$-progressively measurable and $[0,K]$-valued processes.
Naturally, an optimal strategy $u^*$ is characterized by $V_t= J_t^{(u^*)}$.\\
\\
Since we cannot observe the two sources $\theta$ and $B_t$ of uncertainty separately, we face a stochastic optimization problem under partial information.
For overcoming this difficulty we are going to derive an observable estimator for the drift parameter by means of filtering theory in the following section.


\section{Filtering theory}
\label{sec:Filtering}
Our aim is to rewrite \eqref{eq:dynX1} as
\begin{equation}
\label{eq:dynX2_1}X_t=x+\int_0^t(\theta_s-u_s) \,ds+\sigma W_t\,,
\end{equation}
where $(\theta_t)_{t \ge 0}$ is an observable estimator for $\theta$ at time $t$ and $(W_t)_{t \ge 0}$ is a Brownian
motion w.r.t. our observation filtration.\\
\begin{remark}
In \citet[p. 225]{liptser2000} the maximum likelihood estimator for the unknown drift parameter in the present situation is given
by $\hat{\theta}_t=\frac{Z_t-z}{t}$. One may notice that this estimator only uses the information which is given by $Z_t$ and does not consider
the whole path up to time $t$. In the following we are going to derive an alternative estimator for $\theta$ which is based on an
application of Bayes' rule. This estimator obeys an appealing representation in terms of a stochastic integral and induces natural boundary conditions for
the optimization problem.
\end{remark}
Using
\[
\P(Z_t \in [\bar{z},\bar{z}+d\bar{z}] | \theta=\theta_j)
=\frac{1}{\sqrt{2\pi}\sigma\sqrt{t}}\exp\left(-\frac{(\bar{z}-\theta_jt-z)^2}{2\sigma^2t}\right)\,dz\,,
\]
and the law of total probability we arrive at
\begin{align*}
\P(Z_t \in [\bar{z},\bar{z}+d\bar{z}])
=\left( q\frac{1}{\sqrt{2\pi}\sigma\sqrt{t}}\exp\left(-\frac{(\bar{z}-\theta_1t-z)^2}{2\sigma^2t}\right)
+(1-q) \frac{1}{\sqrt{2\pi}\sigma\sqrt{t}}\exp\left(-\frac{(\bar{z}-\theta_2t-z)^2}{2\sigma^2t}\right)\right)\,dz\,.
\end{align*}
Now Bayes' theorem allows us to determine
\begin{align*}
\P(\theta=\theta_1 \, \vline \, Z_t \in [\bar{z},\bar{z}+d\bar{z}])
&=\frac{q\frac{1}{\sqrt{2\pi}\sigma\sqrt{t}}\exp\left(-\frac{(\bar{z}-\theta_1t-z)^2}{2\sigma^2t}\right)\,dz}{\left( q\frac{1}{\sqrt{2\pi}\sigma\sqrt{t}}\exp\left(-\frac{(\bar{z}-\theta_1t-z)^2}{2\sigma^2t}\right)
+(1-q) \frac{1}{\sqrt{2\pi}\sigma\sqrt{t}}\exp\left(-\frac{(\bar{z}-\theta_2t-z)^2}{2\sigma^2t}\right)\right)\,dz
}\\
&=\frac{1}{1
+\frac{1-q}{q}\exp\left(\frac{(\theta_2-\theta_1)(\bar{z}-z-\frac{1}{2}(\theta_1+\theta_2)t)}{\sigma^2}\right)
}\,,
\end{align*}
and to compute $\theta_t=\E(\theta \, \vline \, Z_t \in [\bar{z},\bar{z}+d\bar{z}])$.
Setting $\theta_t=h(t,Z_t)$, we finally get
\begin{align*}
h(t,\bar{z})
=\theta_1 \frac{1}{1+f(t,\bar{z})}+\theta_2\frac{f(t,\bar{z})}{1+f(t,\bar{z})}
=\theta_1 +(\theta_2-\theta_1)\frac{f(t,\bar{z})}{1+f(t,\bar{z})}\,,
\end{align*}
\begin{align*}
f(t,\bar{z})=\frac{1-q}{q}\exp\left(\frac{(\theta_2-\theta_1)(\bar{z}-z-\frac{1}{2}(\theta_1+\theta_2)t)}{\sigma^2}\right)\,.
\end{align*}
As a consequence we can state the following Lemma.
\begin{lemma}\label{lem:Ztheta}
$Z_t$ and $\theta_t$ are connected via
\be\label{eq:zexpl}
Z_t=\frac{\sigma^2\log\frac{(\theta_t-\theta_1)q}{(\theta_2-\theta_t)(1-q)}}{(\theta_2-\theta_1)}+z+\frac{1}{2}(\theta_1+\theta_2)t\,.
\ee
In particular, $(\theta_t)_{t \ge 0}$ is adapted to the observation filtration.
\end{lemma}

From \citet[Theorem 9.1]{liptser1977} we get that
\[
W_t=\frac{1}{\sigma} \left( \theta t - \int_0^t \theta_s \,ds \right) +B_t
\]
is an $\{\cF^{Z}_t\}_{t\ge 0}$-Brownian motion, sometimes referred to as innovation process.
Therefore we can rewrite \eqref{eq:dynX1} and \eqref{eq:dynZ1} as
\begin{align}
\label{eq:dynX2_2}X_t&= x + \int_0^t (\theta_s-u_s) \,ds +\sigma W_t\,,\\
\label{eq:dynZ2}Z_t&= z + \int_0^t \theta_s \,ds +\sigma W_t\,.
\end{align}
Due to \eqref{eq:zexpl}, $\cF^{Z}_t=\cF^{X,\theta}_t$ for all $t\ge 0$, 
where $\{\cF^{X,\theta}_t\}_{t\ge 0}$ is the augmented filtration
generated by $X$ and $(\theta_t)_{t\ge0}$.
Using It\^o's formula we derive \eqref{eq:dyntheta}, such that from now on we can consider the following system of state variables
\begin{align}
\label{eq:dynX2}X_t&= x + \int_0^t (\theta_s-u_s) \,ds +\sigma W_t\,,\\
\label{eq:dyntheta}\theta_t&= \vartheta + \frac{1}{\sigma}\int_0^t (\theta_s-\theta_1)(\theta_2-\theta_s) \,dW_s\,.
\end{align}

\begin{remark}
Let $u$ be any progressively measurable bounded process.
Equation \eqref{eq:dyntheta} does not depend on $u$ and therefore has a
solution by the well-known theorem on existence and uniqueness of solutions
of SDEs with Lipschitz coefficients. See, for example, 
\citet[Chapter 2, Theorem 7]{krylov1980}.  

$X$ does not appear on the right hand side of \eqref{eq:dynX2}, so 
this becomes just an ordinary integral.
\end{remark}



\section{The Hamilton-Jacobi-Bellman equation}
\label{sec:HJB}
In this section we will show that the dynamic programming approach is applicable for solving the optimization problem when considering
$(X_t,\theta_t)$ as state variables.\\
Remember the definition of the optimal value process
\[
V_t
:=\sup_{u\in A}\E\left(\int_t^{\tau}e^{-\delta(s-t)}u_s \,ds \, \vline \, \cF_t^{Z}\right)
=\sup_{u\in A}\E\left(\int_t^{\tau}e^{-\delta(s-t)}u_s \,ds \, \vline \, \cF_t^{X,\theta}\right)\,.
\]
The system \eqref{eq:dynX2} and \eqref{eq:dyntheta} describes autonomous
state dynamics in the sense of \cite[Section IV.5]{fleming2006}.
It is therefore 
natural to consider Markov controls in the following.  Furthermore, due to the
Markovian structure and the infinite horizon, 
we
get that
\begin{align*}
V_t&
=\sup_{u\in A}\E\left(\int_t^{\tau}e^{-\delta(s-t)}u_s \,ds \, \vline \, \cF^{X,\theta}_t\right)
=\sup_{u\in A}\E\left(\int_t^{\tau}e^{-\delta(s-t)}u_s \,ds \, \vline \, X_t,\theta_t\right)
=V(X_t,\theta_t)\,, 
\end{align*}
a.s., where from now on $\tau:=\inf\{t\ge 0 \, \vline \, X_t\le 0\}$ and 
\[
V(x,\vartheta):=\sup_{u\in A}\E\left(\int_0^{\tau}e^{-\delta s}u_s \,ds \, \vline \, X_0=x, \theta_0=\vartheta\right)
\]
denotes the optimal value function of the optimization problem.
For a strategy $u\in A$ we define 
\[
J^{(u)}(x,y):=\E\left(\int_0^{\tau}e^{-\delta s}u_s \,ds \, \vline \, X_0=x, \theta_0=\vartheta\right)\,.
\]
From now on we abbreviate the expectation given the initial values $X_0=x$ and $\theta_0=\vartheta$ by $\E_{x,\vartheta}$.

\begin{remark}
From \citet[Chapter 3, Theorem 5]{krylov1980} we know that the 
optimal value function $V$ is continuous. 
\end{remark}

\begin{lemma}\label{th:x-infinity}
We have $0\le V \le \frac{K}{\delta}$ and 
$\lim_{x \to \infty} V(x,\vartheta)=\frac{K}{\delta}$ uniformly in $\vartheta$.
Furthermore, $V$ is increasing in both parameters.
\end{lemma}

\begin{proof}
Clearly, $0\le V(x,\vartheta)\le \int_0^\infty K e^{-\delta s}ds=\frac{K}{\delta}$.
On the other hand, we have 
$V(x,\vartheta)\ge \E_{x,\vartheta}(\int_0^\tau e^{-\delta t} K \,dt)$,
where 
\begin{align*}
(X_0,\theta_0)&=(x,\vartheta)\,,\\
dX_t&=(\theta_t-K)dt+\sigma dW_t\,,\\
d\theta_t&=\frac{1}{\sigma}(\theta_t-\theta_1)(\theta_2-\theta_t)dW_t\,,\\
\tau&=\inf\{t\ge 0 \, \vline \, X_t \le 0\}\,.
\end{align*}
Since we always assume $\theta_1<\theta_2$,
\[
X_t=x+\int_0^t (\theta_s-K)\,ds+\sigma W_t
\ge x+\int_0^t (\theta_1-K)\,ds+\sigma W_t=:X^{\min}_t
\]
and therefore $\tau^{\min}:=\inf\{t\ge 0 \, \vline \, X^{\min}_t \le 0\}\le \tau$
such that 
\[
V(x,\vartheta)\ge \E_{x,\vartheta}\left(\int_0^{\tau^{\min}} e^{-\delta t} K \,dt\right)
=\frac{K}{\delta}\left(1-\E_{x,\vartheta}\left(e^{-\delta \tau^{\min}}\right)\right)\,.
\]
The last expectation can be computed using standard techniques:
for every $\lambda>0$ the process defined by $M_t:=e^{-\lambda W_t-\frac{\lambda^2}{2}t}$ is a martingale and the stopped 
process $M^{\tau^{\min}}$ is a bounded martingale, such that 
$\E_{x,\vartheta}(e^{\lambda \frac{x+(\theta_1-K) \tau^{\min}}{\sigma}-\frac{\lambda^2}{2}{\tau^{\min}}})=\E_{x,\vartheta}(e^{-\lambda W_{\tau^{\min}}-\frac{\lambda^2}{2}{\tau^{\min}}})= 1$,
and hence
$\E_{x,\vartheta}(e^{\lambda (\frac{(\theta_1-K) }{\sigma}-\frac{\lambda}{2}){\tau^{\min}}})=e^{-\lambda \frac{x}{\sigma}}$. We therefore get
$ \E_{x,\vartheta}\left(e^{-\delta \tau^{\min}}\right)=e^{-\lambda x}\,$, where
\be\label{eq:lambda}
\lambda:=\frac{(\theta_1-K)+\sqrt{(\theta_1-K)^2+2 \delta \sigma^2}}{\sigma^2}\,.
\ee
The monotonicity of $V$ with respect to both parameters follows from a
pathwise argument similar to that above. Thereby one has to keep in mind
the assumption $\theta_1<\theta_2$.
\end{proof}

Formally applying It\^o's formula to $V$ gives
\begin{align*}
dV(X_{t},\theta_{t})&=
V_xdX_t+\frac{1}{2}V_{xx}(dX_t)^2+V_{x\vartheta}dX_t d\theta_t+V_\vartheta d\theta_t+\frac{1}{2}V_{\vartheta\vartheta}(d\theta_t)^2\\
&=(\theta_t-u_t) V_x dt+\sigma V_x \,dW_t+\frac{\sigma^2}{2}V_{xx} \,dt+(\theta_t-\theta_1)(\theta_2-\theta_t) V_{x\vartheta} \,dt\\
&+ \frac{(\theta_t-\theta_1)(\theta_2-\theta_t)}{\sigma} V_\vartheta \,dW_t+\frac{1}{2 \sigma^2}(\theta_t-\theta_1)^2(\theta_2-\theta_t)^2 V_{\vartheta\vartheta} \,dt\\
&=: \cL V(X_t,\theta_t) \,dt - u_t V_x dt+\left(\sigma V_x + \frac{(\theta_t-\theta_1)(\theta_2-\theta_t)}{\sigma} V_\vartheta \right) \,dW_t\,.
\end{align*}

We now prove a version of the dynamic programming principle, or Bellman principle.
\begin{proposition}[Bellman principle]
For every bounded stopping time $\eta$ we have
\[
V(x,\vartheta)
=\sup_{u \in A} \E_{x,\vartheta}\left(\int_0^{\tau \wedge \eta}e^{-\delta t}u_t\,dt+e^{-\delta(\tau\wedge\eta)}V(X_{\tau\wedge\eta},\theta_{\tau\wedge\eta})\right)\,.
\]
\end{proposition}

\begin{proof}
Let $u^*$ be some $\varepsilon$-optimal strategy for $(x,\vartheta)$, then 
\begin{align*}
V(x,\vartheta)-\varepsilon&<J^{(u^*)}(x,\vartheta) = \E_{x,\vartheta}\left(\int_0^{\tau}e^{-\delta t}u^*_t \,dt\right)\\
&\le \E_{x,\vartheta}\left(\int_0^{\tau \wedge \eta}e^{-\delta t}u^*_t \,dt + e^{-\delta (\tau \wedge \eta)}V(X_{\tau \wedge \eta}, \theta_{\tau \wedge \eta})\right) \nonumber \\
&\le \sup_{u\in A}\E_{x,\vartheta}\left(\int_0^{\tau \wedge \eta}e^{-\delta t}u_t \,dt + e^{-\delta (\tau \wedge \eta)}V(X_{\tau \wedge \eta}, \theta_{\tau \wedge \eta})\right)\,, 
\end{align*}
which proves that $V(x,\vartheta)\le\sup_{u\in A}\E_{x,\vartheta}\left(\int_0^{\tau \wedge \eta}e^{-\delta t}u_t \,dt + e^{-\delta (\tau \wedge \eta)}V(X_{\tau \wedge \eta}, \theta_{\tau \wedge \eta})\right)$, since $\varepsilon>0$ was arbitrary.

On the other hand, for any $\varepsilon>0$  we have
\begin{align*}
\lefteqn{\sup_{u\in A}\E_{x,\vartheta}\left(\int_0^{\tau \wedge \eta}e^{-\delta t}u_t \,dt + e^{-\delta (\tau \wedge \eta)}V(X_{\tau \wedge \eta}, \theta_{\tau \wedge \eta})\right)}\\
&\le \E_{x,\vartheta}\left(\int_0^{\tau \wedge \eta}e^{-\delta t}u_t^{1,\varepsilon} \,dt
+ e^{-\delta (\tau \wedge \eta)} J^{(u^{2,\varepsilon})}(X_{\tau \wedge \eta}, \theta_{\tau \wedge \eta})\right) + \varepsilon\,,
 \end{align*}
where $u^{1,\varepsilon}$ and $u^{2,\varepsilon}$ are $\frac{\varepsilon}{3}$-optimal
and $\frac{2\varepsilon}{3}$-optimal strategies for $(x,\vartheta)$ and
$(X_{\tau \wedge \eta}, \theta_{\tau \wedge \eta})$, respectively.  From the continuity of
$V$ one can construct these strategies by a similar procedure as stated in
\citet{azcue2005}.\\
The concrete procedure for constructing $u^{2,\varepsilon}$ is as follows. 
Fix $\varepsilon>0$, 
and determine $B \in (0,\infty)$ such that 
$\frac{K}{\delta}-V(x,\vartheta)<\frac{2\varepsilon}{3}$ for all 
$x\ge B$. Note that for an initial value $(x,\vartheta)$ with $x\ge B$ the strategy
$\tilde u\equiv K$ is $\frac{2\varepsilon}{3}$-optimal.

Choose grid points $x_1<\ldots<x_N$ and $\vartheta_1< \ldots < \vartheta_N$
such that the rectangles $[x_i,x_{i+1}]\times[\vartheta_j,\vartheta_{j+1}]$
cover $[0,B]\times[\theta_1,\theta_2]$. 
Since $V$ is increasing with each parameter, we have that
\begin{align*}
V(x_i,\vartheta_j)\le V(x,\vartheta)\le V(x_{i+1},\vartheta_{j+1})\,,
\end{align*}
whenever $(x,\vartheta)\in[x_i,x_{i+1}]\times[\vartheta_j,\vartheta_{j+1}]$.
Now, because of continuity of $V$, the number $N$ and the grid points can be 
chosen such that 
\begin{align*}
 0\leq V(x_{i+1},\vartheta_{j+1})-V(x_{i},\vartheta_{j})\leq\frac{\varepsilon}{3}\,,
\end{align*}
for all $i,j$. 
Let $\tilde{u}_{ij}$ be an 
$\frac{\varepsilon}{3}$-optimal strategy for $(x_i,\vartheta_j)$,
\begin{align*}
 V(x_i,\vartheta_j)\leq J^{(\tilde{u}_{ij})}(x_i,\vartheta_j)+\frac{\varepsilon}{3}\,.
\end{align*}
Note that also $
J^{(\tilde{u}_{ij})}(x_i,\vartheta_j)\le J^{(\tilde{u}_{ij})}(x,\vartheta)\le J^{(\tilde{u}_{ij})}(x_{i+1},\vartheta_{j+1})$.

The strategies $\tilde u^{ij}$ together with $\tilde u \equiv K$ define the strategy $u^{2,\varepsilon}$ for $(X_{\eta\wedge\tau},\theta_{\eta\wedge\tau})$.
For $(X_{\eta\wedge\tau},\theta_{\eta\wedge\tau})\in[x_i,x_{i+1}]\times[\vartheta_j,\vartheta_{j+1}]$ we get
\begin{align*}
0&\le V(X_{\eta\wedge\tau},\theta_{\eta\wedge\tau})-J^{(\tilde{u}_{ij})}(X_{\eta\wedge\tau},\theta_{\eta\wedge\tau})
\le V(x_{i+1},\vartheta_{j+1})-J^{(\tilde{u}_{ij})}(x_i,\vartheta_j)\\
&\le V(x_{i},\vartheta_{j})-J^{(\tilde{u}_{ij})}(x_i,\vartheta_j)+\frac{\varepsilon}{3}\le \frac{2\varepsilon}{3}\,.
\end{align*}
Finally, we choose an $\frac{\varepsilon}{3}$-optimal strategy $u^{1,\varepsilon}$ for the right hand side of the dynamic programming principle, i.e.,
\[
\sup_{u\in A}\E_{x,\vartheta}\left(\int_0^{\tau\wedge \eta}e^{-\delta t} u_t dt\right)
\le\E_{x,\vartheta}\left(\int_0^{\tau\wedge \eta}e^{-\delta t} u^{1,\varepsilon}_t dt\right)+\frac{\varepsilon}{3}.
\]
Now we put everything together. For any strategy $u$
\begin{align*}
\lefteqn{\E_{x,\vartheta}\left(\int_0^{\tau\wedge \eta}e^{-\delta t} u_t dt
+e^{-\delta(\tau\wedge\eta)}V(X_{\tau\wedge\eta},\theta_{\tau\wedge\eta}) \right)}\\
&=\E_{x,\vartheta}\left(\int_0^{\tau\wedge \eta}e^{-\delta t} u_t dt\right)
+\E_{x,\vartheta}\left(e^{-\delta(\tau\wedge\eta)}V(X_{\tau\wedge\eta},\theta_{\tau\wedge\eta}); X_{\tau\wedge\eta}\ge B \right)\\
&\quad+\sum_{i,j}\E_{x,\vartheta}\left(e^{-\delta(\tau\wedge\eta)}V(X_{\tau\wedge\eta},\theta_{\tau\wedge\eta}); (X_{\tau\wedge\eta},\theta_{\tau\wedge\eta})\in [x_{i},x_{i+1})\times [\vartheta_j,\vartheta_{j+1}] \right)\\ 
&\le\E_{x,\vartheta}\left(\int_0^{\tau\wedge \eta}e^{-\delta t} u^{1,\varepsilon}_t dt\right)
+\frac{\varepsilon}{3}
+\E_{x,\vartheta}\left(e^{-\delta(\tau\wedge\eta)} J^{(\tilde u)}(X_{\tau\wedge\eta},\theta_{\tau\wedge\eta}); X_{\tau\wedge\eta}\ge B \right)\\
&\quad+\sum_{i,j}\E_{x,\vartheta}\left(e^{-\delta(\tau\wedge\eta)}J^{(\tilde u_{ij})}(X_{\tau\wedge\eta},\theta_{\tau\wedge\eta}); (X_{\tau\wedge\eta},\theta_{\tau\wedge\eta})\in [x_{i},x_{i+1})\times [\vartheta_j,\vartheta_{j+1}] \right)+\frac{2\varepsilon}{3}\\ 
&=J^{(u^\varepsilon)}+\varepsilon\le V(x,\vartheta)+\varepsilon\,,
\end{align*}
where \vspace{-1em}
\begin{align*}
  u_t^\varepsilon:=\begin{cases}
               u_t^{1, \varepsilon}, & t \le \tau \wedge \eta\\
               u_t^{2, \varepsilon}, & t > \tau \wedge \eta\,.
              \end{cases}
\end{align*}
Thus we have for any $u$ that
$\E_{x,\vartheta}\left(\int_0^{\tau\wedge \eta}e^{-\delta t} u_t dt
+e^{-\delta(\tau\wedge\eta)}V(X_{\tau\wedge\eta},\theta_{\tau\wedge\eta}) \right)<V(x,\vartheta)+\varepsilon$,
and therefore
\[
\sup_{u\in A}\E_{x,\vartheta}\left(\int_0^{\tau\wedge \eta}e^{-\delta t} u_t dt
+e^{-\delta(\tau\wedge\eta)}V(X_{\tau\wedge\eta},\theta_{\tau\wedge\eta}) \right)\le V(x,\vartheta)+\varepsilon\,.
\]
Since $\varepsilon > 0$ was arbitrary, the proof is finished.
\end{proof}


Now, still under the assumption $V \in C^2$, we apply the dynamic programming principle and derive
\begin{align*}
V(x,\vartheta)
&=\sup_{u \in A}\E_{x,\vartheta}\left(\int_0^{\tau}e^{-\delta t}u_t \,dt\right)
=\sup_{u \in A}\E_{x,\vartheta}\left(\int_0^{\tau \wedge \eta}e^{-\delta t}u_t\,dt+e^{-\delta (\tau \wedge \eta)}V(X_{\tau \wedge \eta},\theta_{\tau \wedge \eta})\right)\\
&=\sup_{u \in A}\E_{x,\vartheta}\left(\int_0^{\tau \wedge \eta}e^{-\delta t}u_t\,dt+e^{-\delta (\tau \wedge \eta)}\left(V(x,\vartheta)+\int_0^{\tau \wedge \eta}\!\!\!\!dV(X_t,\theta_t)\right)\right)\\
&=\sup_{u \in A}\E_{x,\vartheta}\left(\int_0^{\tau \wedge \eta}e^{-\delta t}u_t\,dt+e^{-\delta (\tau \wedge \eta)}\left(V(x,\vartheta)
+\int_0^{\tau \wedge \eta}\!\!\!\!\cL V(X_t,\theta_t)-V_x u_t \,dt\right)\right)\,.
\end{align*}
Therefore,
\begin{align*}
\frac{1-e^{-\delta (\tau \wedge \eta)}}{\tau \wedge \eta}V(x,\vartheta)=
\sup_{u \in A}\E_{x,\vartheta}\left(\frac{1}{\tau \wedge \eta}\int_0^{\tau \wedge \eta}e^{-\delta t}u_t\,dt
+\frac{1}{\tau \wedge \eta}\int_0^{\tau \wedge \eta}\!\!\!\!\cL V(X_t,\theta_t)-V_x u_t \,dt\right)\,,
\end{align*}
and by letting $\eta\rightarrow 0$ we arrive at
\[
\delta V=\sup_{u \in [0,K]} \left(u+\cL V- u V_x \right)\,.
\]
Thus, the associated Hamilton-Jacobi-Bellman equation is given by
\begin{align}
 \label{eq:HJB}(\cL - \delta) V +\sup_{u \in [0,K]}(u(1-V_x))=0\,.
\end{align}
The HJB equation is a second order degenerate-elliptic PDE since
there is only one Brownian motion driving the two-dimensional process $(X_t,\theta_t)_{t \ge 0}$.\\
Now it remains to find appropriate boundary conditions. The ones for 
$x=0$ and $x\rightarrow \infty$ follow from the definition of $V$ and from 
Lemma \ref{th:x-infinity}, respectively:
\begin{align}
\label{eq:BC1}V(0,\vartheta)&=0\,,\\
\label{eq:BC2}V(B,\vartheta)&\rightarrow\frac{K}{\delta} \mbox{ uniformly in }\vartheta\mbox{ for }B \rightarrow \infty\,.
\end{align}
The ones for $\vartheta\rightarrow\theta_i$, $i=1,2$ are obtained by solving
the optimal control problem for known deterministic drift, as has been 
done in \citet{asmussen1997}. We give their solution with
notation adapted to our setup:
\begin{align}
\label{eq:BC3}V(x,\theta_i)&=\begin{cases}
                                  a_{1,i} \exp(\alpha_{1,i}(x-\bar{b}_i))+a_{2,i} \exp (-\alpha_{2,i}(x-\bar{b}_i)), & x<\bar{b}_i\\
                                  b_{2,i} \exp (-\beta_{2,i} (x-\bar{b}_i)) + \frac{K}{\delta}, & x \ge \bar{b}_i \,,
                                 \end{cases}
\end{align}
where $i \in \lbrace1,2 \rbrace$ and 
\begin{align}
\label{eq:classical_pars1}
\alpha_{1,i}&=\frac{1}{\sigma^2} \left( - \theta_i + \sqrt{\theta_i^2 + 2 \sigma^2 \delta} \right)\,,\\
\label{eq:classical_pars2}\alpha_{2,i}&= \frac{1}{\sigma^2} \left(\theta_i + \sqrt{\theta_i^2 + 2 \sigma^2 \delta} \right)\,,\\
\label{eq:classical_pars3}\beta_{2,i}&= \frac{1}{\sigma^2} \left( \theta_i-K + \sqrt{\left(\theta_i-K\right)^2 + 2 \sigma^2 \delta} \right)\,,\\
\label{eq:classical_pars5}a_{1,i}&=\frac{\alpha_{2,i}(\frac{K}{\delta}-\frac{1}{\beta_{2,i}})+1}{\alpha_{1,i}+\alpha_{2,i}}\,,\\
\label{eq:classical_pars6}a_{2,i}&=\frac{\alpha_{1,i}(\frac{K}{\delta}-\frac{1}{\beta_{2,i}})-1}{\alpha_{1,i}+\alpha_{2,i}}\,,\\
\label{eq:classical_pars4}\bar{b}_i&=\left(\frac{1}{\alpha_{1,i}+\alpha_{2,i}}\log\left(-\frac{a_{1,i}}{a_{2,i}}\right)\right)_+\,,\\
\label{eq:classical_pars7}b_{2,i}&
=\left\{
\begin{array}{cc}
-\frac{1}{\beta_{2,i}}&\bar b_i>0\\
-\frac{K}{\delta}&\bar b_i= 0\,.
\end{array}
\right.
\end{align} 
\begin{remark}
We could give the parameters $\theta_1,\theta_2,\delta,K$
relative to $\sigma^2$, that is $\theta_i=\hat \theta_i\sigma^2$,
$\delta=\hat \delta \sigma^2$, $K=\hat K\sigma^2$. Then $\sigma^2$
cancels from all expressions in \eqref{eq:classical_pars1}--\eqref{eq:classical_pars7}.

In other words: the qualitative behavior of the model only depends 
on the relative values
between the parameters $\theta_1,\theta_2,\delta,K$.
\end{remark}
\begin{remark}
For known deterministic drift $\theta_i$ the optimal strategy is of
threshold type: no dividends are paid for $X_t<\bar b_i$, and for $X_t\ge \bar
b_i$ dividends are paid at maximum rate $K$.

A numerical solution of the HJB equation will be presented in Section
\ref{sec:Num}. We will see that the numerical results suggest that also
in our Bayesian setup a threshold strategy is optimal.
\end{remark}


\section{Viscosity Solution Characterization}
\label{sec:Verif}

In this section we  present the main theoretical results of this paper.\\
In the univariate setting, as described in \citet{asmussen1997},
it is possible to determine a smooth explicit solution of the associated HJB equation,
whereas in the present situation the HJB equation \eqref{eq:HJB} hardly allows for such a solution.
As a consequence one needs to rely on numerical methods for obtaining a solution of the optimization problem
and the crucial theoretical basis is the uniqueness of a solution of \eqref{eq:HJB}.
Since a-priori the regularity of a solution is questionable, we need a weaker solution concept, which still allows to prove uniqueness.\\

Therefore, we characterize the optimal value function $V$ as the unique viscosity solution of \eqref{eq:HJB},
since this concept also serves as a basis for numerical considerations.
  \\
Below we present the concept of viscosity solutions for the HJB equation under study.
A more detailed treatment can be found in \citet{fleming2006} or \citet{crandall1992}.
Denote 
$\Omega:=(0,\infty) \times (\theta_1, \theta_2)$, $\bar{\Omega}=[0,\infty)\times [\theta_1, \theta_2]$ and let $\partial \bar{\Omega}$
denote its boundary.
\begin{definition}
(viscosity solution)
\begin{enumerate}
 \item A function $w:\bar{\Omega} \rightarrow \R$ is a {\it viscosity subsolution} to \eqref{eq:HJB} if
 \[
  -\delta \phi(\bar{x},\bar{\vartheta}) + \cL \phi(\bar{x},\bar{\vartheta}) + \sup_{u \in [0,K]}(u(1-\phi_x(\bar{x},\bar{\vartheta})))\ge 0
 \]
 for all $(\bar{x},\bar{\vartheta}) \in \Omega$ and for all $\phi \in C^2(\bar{\Omega})$ such that $(\bar{x},\bar{\vartheta})$ is a maximum of $w-\phi$ with $w(\bar{x},\bar{\vartheta})=\phi(\bar{x},\bar{\vartheta})$.
 
 \item A function $w:\bar{\Omega} \rightarrow \R$ is a {\it viscosity supersolution} to \eqref{eq:HJB} if
 \[
  -\delta \psi(\bar{x},\bar{\vartheta}) + \cL \psi(\bar{x},\bar{\vartheta}) + \sup_{u \in [0,K]}(u(1-\psi_x(\bar{x},\bar{\vartheta})))\le 0
 \]
 for all $(\bar{x},\bar{\vartheta}) \in \Omega$ and for all $\psi \in C^2(\bar{\Omega})$ such that $(\bar{x},\bar{\vartheta})$ is a minimum of $w-\psi$ with $w(\bar{x},\bar{\vartheta})=\psi(\bar{x},\bar{\vartheta})$.
 
\item $w:\bar{\Omega} \rightarrow \R$ is a {\it viscosity solution} to \eqref{eq:HJB}  if it is both a viscosity sub- and supersolution.
\end{enumerate}
\end{definition}
The derivation of \eqref{eq:HJB} was done in a heuristic way.
The following theorem shows that the optimal value function $V$ is indeed connected to it in a weak sense.
\begin{theorem}\label{thm:viscosity}
$V(x,\vartheta)$ is a viscosity solution of \eqref{eq:HJB} with boundary conditions \eqref{eq:BC1}, \eqref{eq:BC2} and \eqref{eq:BC3}.
\end{theorem}
The uniqueness result is based on comparison, which is dealt with in the next theorem.\\
Remember that the optimal value function $V$ exhibits the following properties:
it is continuous on $\bar{\Omega}$, 
bounded, i.e., $0\le V\leq
\frac{K}{\delta}$, and
$\lim_{x\to\infty}V(x,\vartheta)=\frac{K}{\delta}$ uniformly in $\vartheta$.

\begin{theorem}[Comparison]\label{thm:comparison}
Let $w$ and $v$ be a bounded and continuous viscosity subsolution and supersolution of \eqref{eq:HJB},
respectively.\\ 
If $w\leq v$ on
$\partial\bar{\Omega}$ and 
$\limsup_{x\rightarrow\infty}(w-v)(x,\vartheta)\le 0$ uniformly in $\vartheta$, then $w\leq v$ on $\Omega$.
\end{theorem}

As a corollary we get uniqueness of the viscosity solution of \eqref{eq:HJB}.
\begin{cor}
The optimal value function $V$ is the unique bounded viscosity solution of \eqref{eq:HJB} with the given boundary conditions.
\end{cor}
\begin{proof}
Suppose there is another solution $W$. Then since both $W$ and $V$ are subsolutions and supersolutions fulfilling the same boundary
conditions we get that $W\leq V$ and $V\leq W$.
\end{proof}

Finally we give a verification theorem. 

\begin{theorem}\label{thm:verification}
Let $v(x,\vartheta)$ be a viscosity supersolution of \eqref{eq:HJB} with
boundary conditions \eqref{eq:BC1}, \eqref{eq:BC2}, \eqref{eq:BC3}, and let $v \in C^2$ almost everywhere. Then $V \le v$.
\end{theorem}

\begin{remark}
Suppose one can construct a strategy $\tilde{u}$ such that $J^{(\tilde{u})}$ is a supersolution with $J^{(\tilde{u})} \in C^2$ almost everywhere.
Then Theorem \ref{thm:verification} implies that $J^{(\tilde{u})}=V$ and thus $\tilde{u}=u^*$ is the optimal strategy.
\end{remark}


\section{Policy iteration and numerical examples}
\label{sec:Num}

In this section we describe a numerical scheme to compute an approximation 
to the optimal dividend policy and optimal value function. 

We have already noted in Lemma \ref{th:x-infinity} that, 
as $x$ becomes large, the optimal value function approaches 
$\frac{K}{\delta}$ uniformly in $\vartheta$. 
For our algorithm we choose a large number $B$ and approximate the 
domain of the value function by $[0,B]\times [\theta_1,\theta_2]$. (In the
numerical examples $B$ is chosen such that $e^{-\lambda B} \le 0.01$, where
$\lambda$ is defined by Equation \eqref{eq:lambda} .) We impose the
boundary conditions \eqref{eq:BC1}, \eqref{eq:BC3},
and 
\begin{equation}
V(B,\vartheta)=\frac{\theta_2-\vartheta}{\theta_2-\theta_1}
b_{1,1} \exp (\beta_1 (B-\bar{b}_1)) 
+\frac{\vartheta-\theta_1}{\theta_2-\theta_1}b_{1,2} \exp (\beta_2 (B-\bar{b}_2)) + \frac{K}{\delta}\,,
\end{equation}
where $b_{1,i},\bar{b}_i,\beta_i$ are as defined by equations
\eqref{eq:classical_pars3}, \eqref{eq:classical_pars4}, \eqref{eq:classical_pars7}, and we may assume that $B\ge \max(\bar b_1,\bar b_2)$.
We thus have continuous boundary conditions and the ones for $x=B$ differ from
$V$ by less than $\frac{K}{\delta}(1-e^{-\lambda B})$.

Next we define a grid $G:=\{x_0,\ldots,x_n\}\times \{\vartheta_0,\ldots,\vartheta_m\}$ in 
$[0,B]\times [\theta_1,\theta_2]$, 
$0=x_0< \ldots <x_n=B$, $\theta_1=\vartheta_0<\ldots<\vartheta_m=\theta_2$.
We want to be able to put more grid points into regions which are of higher 
interest, that is,
close to the values $\theta_1,\theta_2$ and, in the $x$-direction, between 
$\bar b_1$ and $\bar b_2$. More concretely, we choose bijective  
$C^1$ functions $h_1,h_2:[0,1]\longrightarrow[0,1]$ and we set $x_k=h_1(j/n)$, $k=0,\ldots,n$
and $\vartheta_k = h_2(k/m)$, $k=0,\ldots,m$.

We use the following definitions, with $0<a_1<a_2<1$, $0\le a_3\le 1$:
\[
h_1(z)=\begin{cases}
\frac{c_2-c_1}{a_2-a_1}z+\frac{a_1^2- (a_1-z)^2}{a_1^2}\left(c_1-\frac{c_2-c_1}{a_2-a_1}a_1 \right)\,,
&0\le z<a_1\\
c_1+\frac{c_2-c_1}{a_2-a_1}(z-a_1)\,,
&a_1\le z<a_2\\
c_2+\frac{c_2-c_1}{a_2-a_1}(z-a_2)+\left(B-c_2-\frac{c_2-c_1}{a_2-a_1}(1-a_2)\right)\frac{(z-a_2)^2}{(1-a_2)^2}\,,
&a_2\le z \le 1\,,
\end{cases}
\]
where $c_1:=\min(\bar b_1,\bar b_2)$, $c_2:=\max(\bar b_1,\bar b_2)$,
and 
\[
h_2(y)=
\theta_1+(\theta_2-\theta_1)((1-a_3)y+a_3(3 y^2-2 y^3))\,.
\]
Note that $h_1,h_2$ are continuously differentiable with $h_j'>0$ in $(0,1)$ and
that we get an evenly spaced grid for $a_1=\bar b_2$, $a_2=\bar b_1$, $a_3=0$.

We start with a simple (threshold) strategy: let 
$b(\vartheta):=\frac{\theta_2-\vartheta}{\theta_2-\theta_1}
\bar{b}_1 
+\frac{\vartheta-\theta_2}{\theta_2-\theta_1}\bar{b}_2$, 
and consider the Markov strategy 
$u^{(0)}(x,\vartheta)=K1_{\{x\ge b(\vartheta)\}}$, $(x,\vartheta)\in G$, i.e., dividends are paid at the maximum intensity $K$, if the firm value exceeds the threshold level $b(\vartheta)$, otherwise no dividends are paid.
We use policy iteration to improve the strategy.

More precisely, if a Markov strategy $u^{(k)}$ is given, we solve
\[
({\cal L}^G-\delta)V+u^{(k)}(1-\cD^G_xV)=0\,,
\]
where ${\cal L}^{G}$ is the operator ${\cal L}$ with differentiation 
operators replaced by suitable finite differences and $\cD^G_x$ is 
a finite difference approximation to differentiation with respect to $x$.
Then $u^{(k+1)}$ is the function that maximizes $u(1-\cD^G_xV)$, that is
$u^{(k+1)}(x,\vartheta)=K1_{\{\cD^G_x V(x,\vartheta)\le 1 \}}$, and the iteration
stops as soon as $u^{(k+1)}=u^{(k)}$.

The details of the method as well as proofs of convergence can be found 
in \citet[Chapter IX]{fleming2006}.

We computed the optimal strategy and the corresponding value function for
 the parameter choice $\sigma = 1$, $\theta_1 =1$, $\theta_2 =2$, 
$\delta =0.5$, and the following values of $K$ and  $B$:
\vspace{2mm}\\
\begin{tabular}{l@{$\;\vline\;$}l l l l}
$K$&$0.2$&$0.67$&$0.9$&$1.5$\\
\hline
B&$2.22$&$3.33$&$4.17$&$7.46$
\end{tabular}\\

In those examples the iteration stops after 3 steps and
the resulting strategy turns out to be a threshold strategy with the threshold
depending on the estimate of $\theta$.

Figure \ref{fig:thresholdplot} shows the threshold level (blue) 
in dependence of 
$\theta$ and, for comparison,  also the corresponding classical threshold 
level (green) from equation \eqref{eq:classical_pars4}. Interestingly,
the difference between these levels can be 
substantial, both quantitatively and qualitatively.


\begin{figure}[ht]
\begin{center}
\includegraphics[scale=0.5]{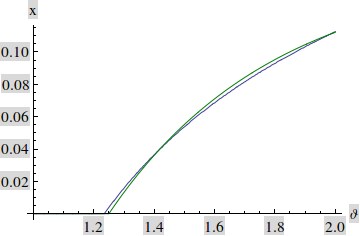}
\includegraphics[scale=0.5]{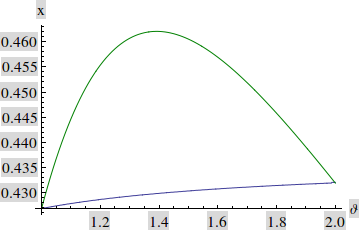}

\includegraphics[scale=0.5]{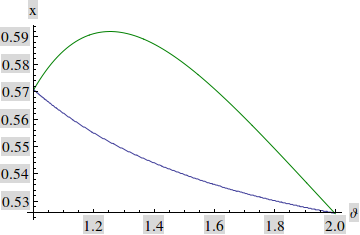}
\includegraphics[scale=0.5]{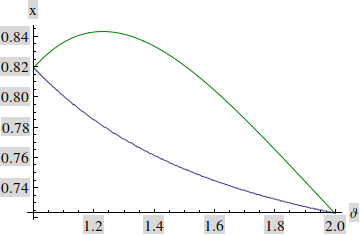}
\end{center}
\caption{The resulting threshold levels for different parameter sets.}\label{fig:thresholdplot}
\end{figure}

Figure \ref{fig:valueplot} shows the value function corresponding to $K=1.5$.
We have refrained from showing the plots of the corresponding 
value functions for the other values of $K$ 
as they all look very similar to the one in 
Figure \ref{fig:valueplot}.

\begin{figure}[ht]
\begin{center}
\includegraphics[scale=0.6]{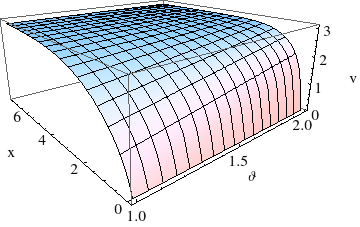}
\end{center}
\caption{The resulting value function.}\label{fig:valueplot}
\end{figure}


\begin{remark}
Figure \ref{fig:valueplot} suggests that the value function is smooth and that therefore it should even be a classical solution of the HJB equation.
However, {\em proving} smoothness is beyond the scope of this paper, since the HJB equation is degenerate elliptic on the whole domain, i.e., the diffusion coefficient is singular, which is highly non-standard.
\end{remark}

\begin{remark}
All of our examples give threshold strategies as the optimal dividend
strategy. The convergence results from \citet[Chapter IX]{fleming2006}
imply that, at least for our parameter sets, we can compute 
an $\varepsilon$-optimal value function corresponding to a threshold strategy.
\end{remark}


\section{Threshold strategies}
\label{sec:Threshold}

The solution for the case where the drift $\theta$ of the uncontrolled wealth
process is deterministic as well as the numerical treatment of the Bayesian
case suggest that the optimal dividend strategy is of threshold type,
that is, there is a threshold level $b$ such that as soon as the wealth process
exceeds the threshold level, dividends are paid at the maximum rate. 

The numerical treatment for the Bayesian case further suggests that 
the threshold level  depends on the estimate for $\theta$. In this section 
we formally define threshold strategies of this type and we give sufficient 
conditions under which they are admissible. We then proceed with giving a sufficient
condition on a threshold strategy for being optimal. 
Of course one would also like to know whether there always
exists at least one optimal
strategy of threshold type. Unfortunately, this question has to remain open
for the time being.

\begin{definition}
Let $b:[\theta_1,\theta_2]\longrightarrow [0,\infty)$ be a continuous function
and let $b(\theta_1)=\bar b_1$, $b(\theta_2)=\bar b_2$,
where $\bar b_1,\bar b_2$ are defined in \eqref{eq:classical_pars4}.

A {\em threshold strategy} with threshold level $b$ is given by
\[
u_t^b=u^b(X_t,\theta_t)=\begin{cases}
                     K, & X_t \ge b(\theta_t)\\
                     0, & X_t < b(\theta_t)\,.
                    \end{cases}
\]
\end{definition}

First, we have to clarify whether threshold strategies are admissible.
The system \eqref{eq:dynX2}, \eqref{eq:dyntheta}, with $u$ replaced by
a stationary Markov strategy, reads
\begin{align}
\label{eq:dynX-Markov}X_t&= x + \int_0^t (\theta_s-u(X_s,\theta_s)) \,ds +\sigma W_t\,,\\
\label{eq:dyntheta-Markov}\theta_t&= \vartheta + \frac{1}{\sigma}\int_0^t (\theta_s-\theta_1)(\theta_2-\theta_s) \,dW_s\,.
\end{align}
A priori it is far from obvious that there exists a solution to the system
\eqref{eq:dynX-Markov}, \eqref{eq:dyntheta-Markov} if $u$ is a general 
measurable function. If $u$ is Lipschitz in both variables, then 
a solution exists by the classical theorem on existence and uniqueness of
solutions of SDEs, see, e.g., \cite[Chapter 2, Theorem 7]{krylov1980}.
But our threshold strategies do not fall in that category.
In \citet{sz14} it is shown that the system \eqref{eq:dynX2}, \eqref{eq:dyntheta} has a unique strong maximal local solution,
if $u=u^b$ is a threshold strategy with a threshold function $b\in C^5$,  
satisfying
\begin{align*}
 \left\| \left(1,-b'(\vartheta)\right) \cdot
 \left( \begin{matrix} \sigma \\ \frac{1}{\sigma} (\theta_2-\vartheta)(\vartheta-\theta_1) \end{matrix} \right)  \right\|^2\ge c > 0
\end{align*}
for some constant $c>0$.
The latter condition means that the diffusion must not be parallel to the discontinuity of the drift. Since in our case $b'$ is a continuous function
on a compact interval, this condition is equivalent to $b'(\vartheta) \neq \frac{\sigma^2}{(\theta_2-\vartheta)(\vartheta-\theta_1)}$.
Furthermore, as the diffusion coefficients are Lipschitz and the drift of $X_t$ is bounded, we can apply \cite[Theorem 3.3]{sz14} and get that the system \eqref{eq:dynX-Markov}, \eqref{eq:dyntheta-Markov}
even has a unique strong global solution until the time $\tau$.
Therefore, threshold strategies are indeed admissible if $b$ is sufficiently regular.
In case of a threshold level like the one 
in Figure \ref{fig:thresholdplot} where for small values of $\vartheta$ we have $b(\vartheta)=0$, and then the curve grows monotonously, $b$ is clearly not sufficiently smooth.
However, in that particular case that does not pose a problem, since $b$ equals
zero at the point of the kink and the process is stopped once it reaches zero,
i.e., when we have ruin. Thus we need not consider solutions starting in
that point or passing through it.

The above discussion is summarized in the following definition and the subsequent lemma:

\begin{definition}
Let $u:[0,\infty)\times[\theta_1,\theta_2]\longrightarrow [0,K]$ be a measurable
function. We call $u$ an {\em admissible Markov strategy} if the system \eqref{eq:dynX-Markov}, \eqref{eq:dyntheta-Markov} has a
strong solution $(X_t,\theta_t)_{t \ge 0}$ on $[0,\tau)$, where, as before, $\tau=\inf\{t\ge
0:X_t\le 0\}$. The set of admissible Markov strategies will be denoted by 
$A_M$.
\end{definition}

\begin{lemma}
Let $b:[\theta_1,\theta_2]\longrightarrow [0,\infty)$ be  a function satisfying
\begin{enumerate}
\item $b$ is continuous;
\item $b(\theta_1)=\bar b_1$, $b(\theta_2)=\bar b_2$;
\item on any interval on which $b>0$ holds, $b$ is $C^5$ and 
$b'(\vartheta) \neq \frac{\sigma^2}{(\theta_2-\vartheta)(\vartheta-\theta_1)}$.
\end{enumerate}
Then $u^b\in A_M$.
\end{lemma}

In the following we give a characterization of value functions corresponding to
optimal threshold strategies.

\begin{definition}
A function $w: \bar{\Omega} \rightarrow \R$ fulfills $w_x \le 1$ 
{\it in the viscosity sense}, if 
for all $\psi \in C^2(\bar{\Omega})$ and 
for all $(\bar{x},\bar{\vartheta}) \in \Omega$ 
such that $(\bar{x},\bar{\vartheta})$ is a minimum of
$w-\psi$ and $w(\bar{x},\bar{\vartheta})=\psi(\bar{x},\bar{\vartheta})$,
we have $\psi_x(\bar{x},\bar{\vartheta}) \le 1$.

In other words, $w_x \le 1$ 
{\it in the viscosity sense}, if
$w$ is a viscosity supersolution of $\psi_x=1$.  
\end{definition}

\begin{definition}
 A function $w: \bar{\Omega} \rightarrow \R$ fulfills ($w_x(x,\vartheta) \le 1 \Leftrightarrow x \ge b(\vartheta)$) {\it in the viscosity sense}, if
 \[
  \psi_x(\bar{x},\bar{\vartheta}) \le 1 \Leftrightarrow \bar{x} \ge b(\bar{\vartheta})
 \]
 for all $(\bar{x},\bar{\vartheta}) \in \Omega$ and for all $\psi \in C^2(\bar{\Omega})$ such that $(\bar{x},\bar{\vartheta})$ is a minimum of $w-\psi$ with $w(\bar{x},\bar{\vartheta})=\psi(\bar{x},\bar{\vartheta})$.
\end{definition}

\begin{remark}
 \begin{enumerate}
\item As for previous definitions of ``viscosity sense'' we have for $w \in C^1$ that [($w_x(x,\vartheta) \le 1 \Leftrightarrow x \ge b(\vartheta)$) in the viscosity sense] $\Leftrightarrow$ [($w_x(x,\vartheta) \le 1 \Leftrightarrow x \ge b(\vartheta)$) $\forall (x,\vartheta) \in \Omega$].
  \item Note that [$\psi_x(x,\vartheta) \le 1 \Leftrightarrow x \ge b(\vartheta)$ $\forall (x,\vartheta) \in \Omega$] $\Leftrightarrow$ [$(1-\psi_x(x,\vartheta))(x-b(\vartheta)) \ge 0$) $\forall (x,\vartheta) \in \Omega$].
 \end{enumerate}
\end{remark}

In the following we denote the value function coming from a threshold strategy 
$u^b\in A_M$ as $J^{(b)}:=J^{(u^b)}$ and\\ $dX_t:=(\theta_t-u_t^b) dt + \sigma dW_t$\,, with $X_0=x$.\\
As in \citet{asmussen1997} one can not guarantee enough smoothness of $J^{(b)}$  for an arbitrary
threshold strategy.
Therefore we characterize $J^{(b)}$ as a viscosity solution of an appropriate PDE.
 
\begin{lemma}
 If the threshold strategy satisfies $u^b \in A_M$, and 
  \[
  J^{(b)}(x,\vartheta)=\E_{x,\vartheta} \left[ \int_0^\tau e^{-\delta t} u^b(X_t,\theta_t) \, dt \right]
 \]
 is continuous, $J^{(b)}$ is a viscosity solution of
  \begin{align}
 \label{eq:HJBb}
  -\delta J^{(b)} + \cL J^{(b)} + K (1-J_x^{(b)}) 1_{\lbrace x \ge b(\vartheta) \rbrace} = 0
 \end{align}
 with the same boundary conditions as for $V$ in \eqref{eq:BC1},\eqref{eq:BC2},\eqref{eq:BC3}.
\end{lemma}

\begin{proof}
 Can be shown using standard techniques similar to the proof of Theorem
\ref{thm:viscosity}.  
To show that  $J^{(b)}$ satisfies the boundary condition \eqref{eq:BC3}, i.e.,
$J^{(b)}$ converges to the non-Bayesian solution for $\theta\rightarrow \theta_i$, $i=1,2$, we use the condition $b(\theta_i)=b_i$.
\end{proof}

The following theorem provides the link between the value of a threshold 
strategy and the HJB equation \eqref{eq:HJB}.
\begin{theorem}\label{th:J-visc-sol}
 Let a threshold level $b:[\theta_1,\theta_2]\longrightarrow [0,\infty)$ exist
with $u^b \in A_M$ and ($J_x^{(b)}(x,\vartheta) \le 1 \Leftrightarrow x \ge b(\vartheta)$) in the viscosity sense.

Then $J^{(b)}$ is a viscosity solution of \eqref{eq:HJB}.
\end{theorem}

\begin{proof}
First, we show that if $J^{(b)}$ is a viscosity subsolution of \eqref{eq:HJBb}, it is also a viscosity subsolution of \eqref{eq:HJB}.\\
Since $J^{(b)}$ is a viscosity subsolution of \eqref{eq:HJBb}, it holds that $\forall \phi \in C^2(\bar{\Omega})$, $\phi \ge J^{(b)}$ in $\Omega$ and for all $(\bar{x},\bar{\vartheta})$ where $\phi(\bar{x},\bar{\vartheta})=J^{(b)}(\bar{x},\bar{\vartheta})$,
 \[
  -\delta \phi + \cL \phi + K (1-\phi_x) 1_{\lbrace \bar{x} \ge b(\bar{\vartheta}) \rbrace} \ge 0\,.
 \]
We have to show that
 \[
  -\delta \phi + \cL \phi + K (1-\phi_x) 1_{\lbrace \phi_x \le 1 \rbrace} \ge 0
 \]
holds in the same points $(\bar{x},\bar{\vartheta})$. Therefore, it is enough  to show that
\[
 -\delta \phi + \cL \phi + K (1-\phi_x) 1_{\lbrace \bar{x} \ge b(\bar{\vartheta}) \rbrace} \le -\delta \phi + \cL \phi + K (1-\phi_x) 1_{\lbrace \phi_x \le 1 \rbrace}\,,
\]
which is equivalent to
\[
 (1-\phi_x) 1_{\lbrace \bar{x} \ge b(\bar{\vartheta}) \rbrace} \le (1-\phi_x) 1_{\lbrace \phi_x \le 1 \rbrace}\,.
\]
If $\phi_x \le 1$, we need that $(1-\phi_x) 1_{\lbrace \bar{x} \ge b(\bar{\vartheta}) \rbrace} \le (1-\phi_x)$, which is obviously true. If $\phi_x > 1$, we need that $(1-\phi_x) 1_{\lbrace \bar{x} \ge b(\bar{\vartheta}) \rbrace} \le 0$. Since $(1-\phi_x) < 0$, this holds, too.\\
So $J^{(b)}$ is a viscosity subsolution of \eqref{eq:HJB}.\\

It remains to show that if $J^{(b)}$ is a viscosity supersolution of \eqref{eq:HJBb}, it is also a viscosity supersolution of \eqref{eq:HJB}.\\
Since $J^{(b)}$ is a viscosity supersolution of \eqref{eq:HJBb}, it holds that $\forall \psi \in C^2(\bar{\Omega})$, $\psi \le J^{(b)}$ in $\Omega$ and for all $(\bar{x},\bar{\vartheta})$ where $\psi(\bar{x},\bar{\vartheta})=J^{(b)}(\bar{x},\bar{\vartheta})$ it holds that
 \[
  -\delta \psi + \cL \psi + K (1-\psi_x) 1_{\lbrace \bar{x} \ge b(\bar{\vartheta}) \rbrace} \le 0\,.
 \]
We have to show that
 \[
  -\delta \psi + \cL \psi + K (1-\psi_x) 1_{\lbrace \phi_x \le 1 \rbrace} \le 0
 \]
holds in the same points $(\bar{x},\bar{\vartheta})$. Hence, we need that
\[
 (1-\psi_x) 1_{\lbrace \bar{x} \ge b(\bar{\vartheta}) \rbrace} \ge (1-\psi_x) 1_{\lbrace \psi_x \le 1 \rbrace}\,.
\]
If $\psi_x < 1$ or $\psi_x > 1$, we get $(1-\psi_x) 1_{\lbrace \bar{x} \ge b(\bar{\vartheta}) \rbrace} \le (1-\psi_x)$ and $(1-\psi_x) 1_{\lbrace \bar{x} \ge b(\bar{\vartheta}) \rbrace} \le 0$, respectively. Therefore, we have to show that $(1-\psi_x) 1_{\lbrace \bar{x} \ge b(\bar{\vartheta}) \rbrace} = (1-\psi_x)$, if $\psi_x < 1$ and $(1-\psi_x) 1_{\lbrace \bar{x} \ge b(\bar{\vartheta}) \rbrace} = 0$, if $\psi_x > 1$. Hence, we need that $\psi_x(\bar{x},\bar{\vartheta}) \le 1 \Leftrightarrow \bar{x} \ge b(\bar{\vartheta})$.\\

From the statement of the theorem we have that  ($J_x^{(b)}(x,\vartheta) \le 1 \Leftrightarrow x \ge b(\vartheta)$) in the viscosity sense,
which is actually equivalent to what we need. Thus, $J^{(b)}$ is a viscosity supersolution of \eqref{eq:HJB}.\\

Altogether, $J^{(b)}$ is a viscosity solution of \eqref{eq:HJB}.
\end{proof}

\begin{cor}
 Let $J^{(b)}$ be like in Theorem \ref{th:J-visc-sol}. Then $J^{(b)}=V$.
\end{cor}

\begin{proof}
 Since $V$ is unique viscosity solution of \eqref{eq:HJB} due to Theorem \ref{thm:comparison}, $J^{(b)}=V$.
\end{proof}

Altogether, we now know that if $u^b \in A_M$ with a threshold level $b$ 
and corresponding value function $J^{(b)}$ such that
$J^{(b)}_x(x,\vartheta) \le 1 \Leftrightarrow x \ge b(\vartheta)$ in the viscosity sense, then $J^{(b)}=V$ and $u^b$ is the optimal control strategy.


\section{Finite time ruin probabilities}
\label{sec:Ruin}

In this section we will determine the finite time ruin probability of the uncontrolled process
\[
 p^Z := \P_{z,\vartheta}(\tau^Z \le t)\,,
\]
where $\tau^Z:= \inf \lbrace s \ge 0 \, \vline \, Z_s \le 0 \rbrace$, for all $t\ge0$.
$\P_{z,\vartheta}$ abbreviates the probability given the initial values $Z_0=z$ and $\theta_0=\vartheta$.\\
Furthermore, we will consider the finite time ruin probability of our
controlled process, assuming that the control variable follows a threshold
strategy as defined in Section \ref{sec:Threshold}. So
\[
 dX_t=\theta_t \,dt + \sigma dW_t - dL_t\,,
\]
where $X_0=x=z$, and $L_t=\int_0^t u_s^b \,ds$.\\
The finite time ruin probability is denoted as
\[
 p^X := \P_{x,\vartheta}(\tau \le t)\,.
\]
Trivially, for $z\le 0$, the finite time ruin probabilities $p^Z$ and $ p^X$ are both equal to 1. Subsequently, we will tacitly assume that $z > 0$.\\

\begin{remark}
For constant and observable $\theta_t=\bar{\theta}$ a classical application of Girsanov's theorem and the reflection principle
yields that the finite time ruin probability of the uncontrolled process $Z$ is given by
 \[
  p^Z_{\bar{\theta}}:=\E \left( \left( 1+e^{-\frac{2\bar{\theta}}{\sigma^2}Z_t} \right) \, 1_{\lbrace Z_t \le 0\rbrace } \right)
= N\left( - \frac{\bar{\theta}t+z}{\sigma \sqrt{t}} \right)+e^{-\frac{2\bar{\theta}z}{\sigma^2}}N\left(\frac{\bar{\theta}t-z}{\sigma \sqrt{t}}\right) \,,
 \]
where $N$ is the cumulative distribution function of the standard normal distribution (cf. \citet[p. 197]{karatzas1991}).
\end{remark}

Now we want to calculate the finite time ruin probability for unobservable $\theta$.
In \citet{elliott2011} results from filtering theory are applied to a general
hidden Markov model and a PDE is derived, the solution of which is proven to be
the finite time survival probability. We apply a different filter to overcome
uncertainty, but after that, we
use a similar result for our processes $Z$ and $X$. From the finite time survival probability we easily get the finite time ruin
probability, which is just the complementary probability.\\
Let
\[
 \cL^Z(\rho(z,\vartheta))=\vartheta \rho_z + \frac{\sigma^2}{2} \rho_{zz} + \frac{1}{2\sigma^2} (\vartheta-\theta_1)^2(\theta_2-\vartheta)^2 \rho_{\vartheta \vartheta} + (\vartheta-\theta_1)(\theta_2-\vartheta) \rho_{z \vartheta}
\]
be the infinitesimal generator of $Z$ and let $\cL^{Z,(0,\infty)}(\rho(z,\vartheta)):=1_{(0,\infty)}(z) \cL^Z(\rho(z,\vartheta))$,
where $1_{(0,\infty)}$ is the indicator function of the domain $(0,\infty)$.\\
Furthermore, let
\[
 \cL^X(\rho(x,\vartheta))=(\vartheta - u^b) \rho_x + \frac{\sigma^2}{2} \rho_{xx} + \frac{1}{2\sigma^2} (\vartheta-\theta_1)^2(\theta_2-\vartheta)^2 \rho_{\vartheta \vartheta} + (\vartheta-\theta_1)(\theta_2-\vartheta) \rho_{x \vartheta}
\]
be the infinitesimal generator of $X$ and let $\cL^{X,(0,\infty)}(\rho(z,\vartheta)):=1_{(0,\infty)}(x) \cL^X(\rho(x,\vartheta))$.\\
Then the following theorem, similar to \citet[Theorem 4.1]{elliott2011}, holds.

\begin{theorem}
\begin{enumerate}
\item If $\Phi^Z(t,z,\vartheta)$ is a smooth solution to
 \begin{align}
  \frac{\partial \Phi^Z}{\partial t}(t,z,\vartheta)=\cL^{Z,(0,\infty)}\left(\Phi^Z(t,z,\vartheta)\right)
 \end{align}
 with initial condition
\[
  \Phi^Z(0,z,\vartheta)=1_{(0,\infty)}(z)
\]
and boundary conditions
\begin{align*}
 \Phi^Z(t,0,\vartheta)&=0\,, \\
 \Phi^Z(t,B,\vartheta)&=1\, \mbox{ for }B \rightarrow \infty\,,
\end{align*}
and $\Phi^Z(t,z,\theta_i)$, $i \in \lbrace \theta_1, \theta_2 \rbrace$ is the solution of the PDE for fixed $\vartheta=\theta_i$, then
\[
 \Phi^Z(t,z,\vartheta)=\P_{z,\vartheta}(\tau^Z > t)\,.
\]

\item If $\Phi^X(t,x,\vartheta)$ is a smooth solution to
 \begin{align}
  \frac{\partial \Phi^X}{\partial t}(t,x,\vartheta)=\cL^{X,(0,\infty)}\left(\Phi^X(t,x,\vartheta)\right)
 \end{align}
 with initial condition
\[
  \Phi^X(0,x,\vartheta)=1_{(0,\infty)}(x)
\]
and boundary conditions
\begin{align*}
 \Phi^X(t,0,\vartheta)&=0\,, \\
 \Phi^X(t,B,\vartheta)&=1\, \mbox{ for }B \rightarrow \infty\,,
\end{align*}
and $\Phi^X(t,x,\theta_i)$, $i \in \lbrace \theta_1, \theta_2 \rbrace$ is the solution of the PDE for fixed $\vartheta=\theta_i$, then
\[
 \Phi^X(t,x,\vartheta)=\P_{x,\vartheta}(\tau > t)\,.
\]
\end{enumerate}
\end{theorem}

\begin{proof}
Similar to the proof of \citet[Theorem 4.1]{elliott2011}.
\end{proof}

\begin{remark}
From the finite time survival probability we easily get the finite time ruin probability
\begin{align*}
  p^Z=1-\Phi^Z(t,z,\vartheta)\,,\\
  p^X=1-\Phi^X(t,x,\vartheta)\,.
\end{align*}
\end{remark}

We solved the PDEs numerically. Note that in our computations the boundary conditions for the $\vartheta$-variable are the numerical solutions of the corresponding PDEs
for fixed $\vartheta=\theta_1$ and $\vartheta=\theta_2$, respectively.

Figure \ref{fig:ruinprob} shows the finite time ruin probability in the uncontrolled and the controlled situation
for the parameter choices $\sigma = 1$, $\theta_1 =1$, $t =10$, and $\theta_2 =2$, $K =0.9$, $B =4.17$, and $\theta_2 =4$, $K =1.5$, $B =7.46$.
The boundary curve coincides with the deterministic situation, where $\theta=\theta_i$, $i=1,2$.
One can see that, depending on the estimate of $\theta$, the difference between the controlled and the uncontrolled situation varies.

\begin{figure}[ht]
 \begin{center}
  \includegraphics[scale=0.3]{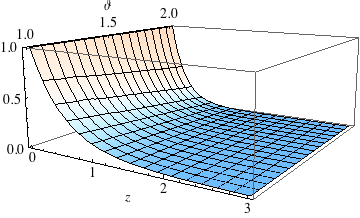}
  \includegraphics[scale=0.3]{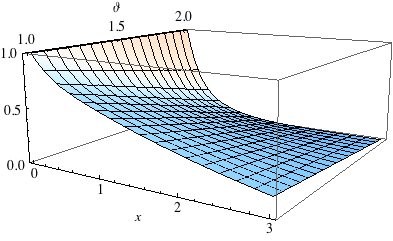}
  \includegraphics[scale=0.3]{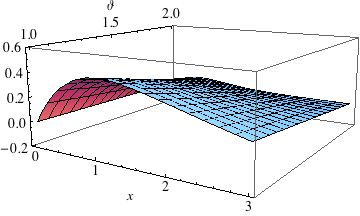}
  
  \includegraphics[scale=0.3]{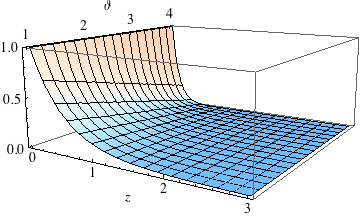}
  \includegraphics[scale=0.3]{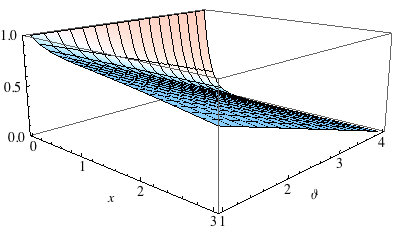}
  \includegraphics[scale=0.3]{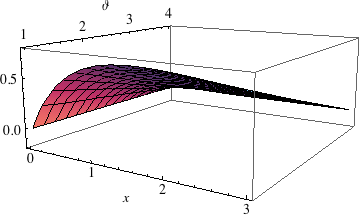}
 \end{center}
 \caption{The finite time ruin probability $p^Z$ (left), $p^X$ (middle), and $p^X-p^Z$ (right).}\label{fig:ruinprob}
\end{figure}

\begin{remark}
 Note that for the uncontrolled process an alternative way to calculate $p^Z$ is
 \begin{align*}
  p^Z=q p^Z_{\theta_1}+(1-q) p^Z_{\theta_2}\,.
 \end{align*}
\end{remark}


\section{Conclusion and open problems}
\label{sec:Concl}

We have presented a dividend optimization problem under partial information
where the drift coefficient of the diffusion firm value process
is a-priori unknown. We have shown how the drift coefficient can be 
estimated and we have derived the HJB equation for the 
stochastic optimal control problem.

It turns out that the optimal value function of the problem is the unique viscosity
solution of the HJB equation, which allows for a numerical treatment of the problem. The numerical method
gives an approximation to the optimal dividend policy and the corresponding
value function. The treated examples suggest that threshold strategies 
are the optimal ones and we have discussed those strategies in more detail.

Finally, we have derived a PDE for the finite time ruin probability in our model
and we have computed concrete examples numerically,
both for the uncontrolled and controlled process.
\\

As already mentioned in Section \ref{sec:Threshold}, a proof that there is
always an optimal strategy of threshold type is yet to be found. The results
of the numerical calculations suggest that this is the case. In addition,
the plots of the corresponding value functions look smooth, so that one
may have the hope that it is actually a classical solution to the
HJB equation.  Due to
the degeneracy of the diffusion coefficient, a proof for this is beyond reach
at the moment.  We formulate the following conjectures:

\begin{conjecture}
The optimal value function $V$ is $C^2$ and is a classical
solution to the HJB equation \eqref{eq:HJB} with boundary conditions 
\eqref{eq:BC1}, \eqref{eq:BC2}, \eqref{eq:BC3}.
Moreover, there always exists an admissible threshold strategy $u^{b}$
such that $V=J^{(b)}$.
\end{conjecture}

\begin{conjecture}
Let $\bar b(\theta)$ denote the optimal threshold for the 
non-Bayesian case
and let $b$ be our optimal threshold function.
From Figure \ref{fig:thresholdplot} one would guess that 
\begin{enumerate}
\item $\min(\bar b(\theta_1),\bar b(\theta_2))\le b(\theta)\le \max(\bar b(\theta_1),\bar b(\theta_2))$;
in particular, if $\bar b(\theta_1)=\bar b(\theta_2)$ then  
$b$ does not depend on the Bayesian estimator;
\item if $0<\bar b(\theta_1)<\bar b(\theta_2)$, then $b$ is strictly increasing and strictly concave;
\item if $\bar b(\theta_1)>\bar b(\theta_2)>0$, then $b$ is strictly decreasing and strictly convex.
\end{enumerate} 
\end{conjecture}

\appendix
\section{Supplementary proofs}

\begin{proof}[Proof of Theorem \ref{thm:viscosity}]
 We have to show that $V$ is both a viscosity sub- and supersolution.\\

\underline{$V$ is a viscosity supersolution:} Let $\psi \in C^2(\bar{\Omega})$, $\psi \le V$ and $(\bar{x},\bar{\vartheta})$ such that $V(\bar{x},\bar{\vartheta})=\psi(\bar{x},\bar{\vartheta})$.\\
  From the dynamic programming principle we have
  \begin{align*}
   \psi(\bar{x},\bar{\vartheta})=V(\bar{x},\bar{\vartheta}) &= \sup_{u \in [0,K]} \E_{\bar{x},\bar{\vartheta}} \left( \int_0^{\tau \wedge \eta} e^{-\delta t} u_t \,dt + e^{-\delta(\tau \wedge \eta)}V(X_{\tau \wedge \eta}, \theta_{\tau \wedge \eta})\right)\\
   &\ge \E_{\bar{x},\bar{\vartheta}} \left( u \frac{1-e^{-\delta(\tau \wedge \eta)}}{\delta} + e^{-\delta(\tau \wedge \eta)} \psi (X_{\tau \wedge \eta}, \theta_{\tau \wedge \eta})\right)
  \end{align*}
  for any fixed $u \in [0,K]$.\\
  Applying It\^o's formula to $\psi$, noting that the stochastic integrals are martingales, dividing by $\eta$ and letting $\eta \rightarrow 0$ gives
  \[
   0 \ge u-\delta \psi (\bar{x},\bar{\vartheta}) + \cL \psi (\bar{x},\bar{\vartheta}) -u \, \psi_x (\bar{x},\bar{\vartheta})\,.
  \]
  Since $u$ was arbitrary we get
  \[
   0 \ge -\delta \psi (\bar{x},\bar{\vartheta}) + \cL \psi (\bar{x},\bar{\vartheta}) + \sup_{u \in [0,K]} \left(u (1-\psi_x (\bar{x},\bar{\vartheta}))\right)\,.
  \]
  So $V$ is a viscosity supersolution.\\
  
  \underline{$V$ is a viscosity subsolution:}  For $\varepsilon > 0$ let $\eta > 0$ with $\varepsilon < \eta^2$.\\
Let $u^{\varepsilon}$ be the density of an $\varepsilon$-optimal dividend
policy, i.e., $J^{(u^{\varepsilon})} \geq V - \varepsilon,$ and denote the firm
value coming from $u^{\varepsilon}$ as $X^{\varepsilon}$. Furthermore, let
$\phi \in C^2 (\bar \Omega)$, $\phi \geq V$ and $(\bar{x}, \bar{\vartheta})$ such that $\phi (\bar{x}, \bar{\vartheta}) = V(\bar{x},
\bar{\vartheta})$.

\begin{align*}
\phi(\bar{x}, \bar{\vartheta}) - \varepsilon &= V(\bar{x}, \bar{\vartheta}) - \varepsilon \le \E_{\bar{x},\bar{\vartheta}} \left( \int_{0}^{\tau \wedge \eta} e^{- \delta t} u_t^{\varepsilon} \, d t + e^{-\delta(\tau \wedge \eta)} V(X_{\tau \wedge \eta}^{\varepsilon}, \theta_{\tau \wedge \eta})\right)\\
& \le \E_{\bar{x},\bar{\vartheta}} \left( \int_0^{\tau \wedge \eta} e^{- \delta t} u_t^{\varepsilon} dt + e^{-\delta(\tau \wedge \eta)} \phi (X_{\tau \wedge \eta}^{\varepsilon}, \theta_{\tau \wedge \eta})\right)\\
&= \E_{\bar{x},\bar{\vartheta}} \left( \int_{0}^{\tau \wedge \eta} e^{-\delta t} u_t^{\varepsilon} dt + e^{- \delta (\tau \wedge \eta)} \left(\phi (\bar{x}, \bar{\vartheta})+ \int_0^{\tau \wedge \eta} \cL \phi \, dt - \int_0^{\tau \wedge \eta} \phi_x u_t^{\varepsilon} \, dt \right)\right)\,,
\end{align*}

where we applied It\^o's formula and used that the stochastic integrals are martingales. Now, we divide by $\eta$ and let $\eta \rightarrow 0$. Since $\varepsilon < \eta^2, \varepsilon \rightarrow 0$ and $\frac{\varepsilon}{\eta} \rightarrow 0$. Thus,
\[
(\cL - \delta) \phi + u^{\varepsilon} (1- \phi_x) \geq 0\,.
\]
Since $u^{\varepsilon}(1-\phi_x) \leq \sup_{u \in [0, K]} u (1- \phi_x)$, we get 
\[
(\cL - \delta) \phi + \sup_{u \in [0, K]} u (1- \phi_x) \geq 0\,.
\]
So V is also a viscosity subsolution.\\

In total $V$ is a viscosity solution of the HJB equation.
\end{proof}

\begin{proof}[Proof of Theorem \ref{thm:comparison}]
We are going to prove the statement of the theorem by contradiction, using standard arguments from \cite{crandall1992} and \cite{pham2009} adapted to our specific situation.\\
Suppose there exists $(x_0,\vartheta_0)\in\bar{\Omega}$ such that
\begin{align*}
w(x_0,\vartheta_0)-v(x_0,\vartheta_0)>0.
\end{align*}
Since $v$ and $w$ are assumed to be bounded, we have
\begin{align*}
0<\sup_{\bar{\Omega}}(w(x,\vartheta)-v(x,\vartheta))=:M<\infty.
\end{align*}
Now on $\partial\bar{\Omega}$ we already have $w\leq v$ by assumption, and 
$\limsup_{x\rightarrow\infty}(w-v)(x,\vartheta)\le 0$ uniformly in $\vartheta$.
Because of that and the  continuity of $w$ and $v$, we have that
the maximum of $w-v$ needs to be attained at an interior point of $\Omega$ 
with finite
$x$-component. 
Therefore there exists $B\in(0,\infty)$ and $(\bar{x},\bar{\vartheta})\in(0,B)\times(\theta_1,\theta_2)$, such that
\begin{align*}
M=w(\bar{x},\bar{\vartheta})-v(\bar{x},\bar{\vartheta}).
\end{align*}
Define for $\alpha>0$,
\begin{align*}
\Phi_{\alpha}(x,\vartheta,y,\rho)=w(x,\vartheta)-v(y,\rho)-\frac{1}{2\alpha}((\vartheta-\rho)^2+(x-y)^2),
\end{align*}
for $(x,\vartheta,y,\rho)\in[0,B]\times[\theta_1,\theta_2]\times[0,B]\times[\theta_1,\theta_2]$.
The function $\Phi_{\alpha}$ is again continuous and it attains a maximum on its compact domain at some point $z_\alpha=(x_\alpha,\vartheta_\alpha,y_\alpha,\rho_\alpha)$.
Furthermore we have
\begin{align*}
0<M&=\Phi_{\alpha}(\bar{x},\bar{\vartheta},\bar{x},\bar{\vartheta})\leq M_\alpha:=\Phi_{\alpha}(z_\alpha)\\
&\leq w(x_\alpha,\vartheta_\alpha)-v(y_\alpha,\rho_\alpha)\leq M.
\end{align*}
The sequence $\{z_\alpha\}_{\alpha>0}$ on $[0,B]\times[\theta_1,\theta_2]\times[0,B]\times[\theta_1,\theta_2]$ is bounded,
therefore there exists a subsequence which converges to some value $\tilde{z}$ when $\alpha\to 0$. At the same time $M_\alpha$ is bounded as well which implies that
\begin{align*}
\frac{1}{2\alpha}((\vartheta_\alpha-\rho_\alpha)^2+(x_\alpha-y_\alpha)^2)
\end{align*}
is bounded as $\alpha\to 0$.  This implies that in the limit $\tilde{z}=(\tilde{x},\tilde{\vartheta},\tilde{x},\tilde{\vartheta})$, and directly from the inequality above we have
$\tilde{x}=\bar{x}$ and $\tilde{\vartheta}=\bar{\vartheta}$. In addition we obtain, at least along another subsequence $M_\alpha\to M$ and
$\frac{1}{2\alpha}((\vartheta_\alpha-\rho_\alpha)^2+(x_\alpha-y_\alpha)^2)\to 0$.\\
Without loss of generality we can assume that we already deal with the convergent subsequence
and, since $\bar{z}$ is an interior point, that $\{z_\alpha\}_{\alpha>0}\in(0,B)\times(\theta_1,\theta_2)\times(0,B)\times(\theta_1,\theta_2)$.\\
\\
In the following step we are going to apply Ishii's Lemma in the form it is stated in \cite[Theorem 3.2]{crandall1992}.
For this purpose we set
\begin{align*}
F^\alpha(x,\vartheta,y,\rho)=\frac{1}{2\alpha}((\vartheta-\rho)^2+(x-y)^2),
\end{align*}
for $(x,\vartheta),(y,\rho)\in[0,B]\times[\theta_1,\theta_2]$.
We have that $w(x,\theta)-v(y,\rho)-F^\alpha(x,\theta,y,\rho)$ attains a maximum in $z_\alpha=(x_\alpha,\vartheta_\alpha,y_\alpha,\rho_\alpha)$. At these points we have
\begin{align*}
D_{x,\theta}F^\alpha(x_\alpha,\vartheta_\alpha,y_\alpha,\rho_\alpha)=-D_{y,\rho}F^\alpha(x_\alpha,\vartheta_\alpha,y_\alpha,\rho_\alpha)=\left(
\begin{array}{c}
\frac{x_\alpha-y_\alpha}{\alpha}\\
\frac{\vartheta_\alpha-\rho_\alpha}{\alpha}                            
\end{array}
\right),
\end{align*}
and with $I_2$ denoting the $2\times 2$ identity matrix we can write
\begin{align*}
D^2 F^\alpha(x_\alpha,\vartheta_\alpha,y_\alpha,\rho_\alpha)=\frac{1}{\alpha}
\left(\begin{array}{c c}
I_2 & -I_2\\
-I_2 & I_2
\end{array}
\right).
\end{align*}
From \cite[Theorem 3.2]{crandall1992} we obtain, for every $\varepsilon>0$, that there
exist symmetric $2\times 2$ matrices $X$ and $Y$ such that
$(D_{x,\theta}F^\alpha(z_\alpha),X)\in \bar{J}^{2,+}_{\bar{\Omega}}w(x_\alpha,\vartheta_\alpha)$, which is the so-called superjet of $w$ at $(x_\alpha,\vartheta_\alpha)$, and 
$(-D_{y,\rho}F^\alpha(z_\alpha),Y)\in \bar{J}^{2,-}_{\bar{\Omega}}v(y_\alpha,\rho_\alpha)$, which is the so-called subjet of $v$ at $(y_\alpha,\rho_\alpha)$.
In particular these matrices fulfill:
\begin{align*}
\left(\begin{array}{c c}
X & 0\\
0 &-Y
\end{array}
\right)\leq \frac{1}{\alpha}
\left(\begin{array}{c c}
I_2 & -I_2\\
-I_2 & I_2
\end{array}
\right) +\frac{\varepsilon}{\alpha^2} \left(\begin{array}{c c}
I_2 & -I_2\\
-I_2 & I_2
\end{array}
\right)^2.
\end{align*}
Choosing $\varepsilon=\alpha$ and taking the square of the matrix explicitly we get
\begin{align}\label{eq:ishii_in}
\left(\begin{array}{c c}
X & 0\\
0 &-Y
\end{array}
\right)\leq \frac{3}{\alpha}
\left(\begin{array}{c c}
I_2 & -I_2\\
-I_2 & I_2
\end{array}
\right)\,.
\end{align}
Before using these super-subjet properties, we are going to derive a bound for the second order terms occurring in the HJB equation.\\
Define $a(x,\vartheta)^T=(\sigma,\frac{(\vartheta-\theta_1)(\theta_2-\vartheta)}{\sigma})$ and write
\begin{align*}
&X=\left(
\begin{array}{cc}
x_{11}  &  x_{12}\\
x_{12}  &  x_{22}                      
\end{array}
\right)\,,\quad Y=\left(
\begin{array}{cc}
y_{11}  &  y_{12}\\
y_{12}  &  y_{22}                      
\end{array}
\right)\,,\\
&\Sigma=\left(
\begin{array}{cc}
a(x_\alpha,\vartheta_\alpha)a(x_\alpha,\vartheta_\alpha)^T  & a(x_\alpha,\vartheta_\alpha)a(y_\alpha,\rho_\alpha)^T \\
a(y_\alpha,\rho_\alpha)a(x_\alpha,\vartheta_\alpha)^T  & a(y_\alpha,\rho_\alpha)a(y_\alpha,\rho_\alpha)^T                        
\end{array}
\right)\,.
\end{align*}
Now we are going to use inequality \eqref{eq:ishii_in},
\begin{align}
\trace&[a(x_\alpha,\vartheta_\alpha)a(x_\alpha,\vartheta_\alpha)^TX-a(y_\alpha,\rho_\alpha)a(y_\alpha,\rho_\alpha)^TY] \nonumber\\
=&x_{11}\sigma^2+2x_{12}(\vartheta_{\alpha}-\theta_1)(\theta_2-\vartheta_{\alpha})+x_{22}\frac{(\vartheta_{\alpha}-\theta_1)^2(\theta_2-\vartheta_{\alpha})^2}{\sigma^2} \nonumber\\
&-y_{11}\sigma^2-2y_{12}(\rho_{\alpha}-\theta_1)(\theta_2-\rho_{\alpha})-y_{22}\frac{(\rho_{\alpha}-\theta_1)^2(\theta_2-\rho_{\alpha})^2}{\sigma^2} \nonumber\\
=&\trace\left[\Sigma\left(\begin{array}{c c}
X & 0\\
0 &-Y
\end{array}
\right)\right]\leq\frac{3}{\alpha}\trace\left[\Sigma\left(\begin{array}{c c}
I_2 & -I_2\\
-I_2 & I_2
\end{array}
\right)\right] \nonumber\\
=&\frac{3}{\alpha}\trace[(a(x_\alpha,\vartheta_\alpha)-a(y_\alpha,\rho_\alpha))(a(x_\alpha,\vartheta_\alpha)-a(y_\alpha,\rho_\alpha) )^T] \nonumber\\
=& \label{eq:help3} \frac{3}{\alpha} (\vartheta_\alpha-\rho_\alpha)^2\frac{(\vartheta_\alpha-\theta_1-\theta_2+\rho_\alpha)^2}{\sigma^2}.
\end{align}
The super-subjet notions appear in an equivalent formulation of the viscosity solution property based on second-order super and subdifferentials,
see \cite[Lemma 4.1, p. 211]{fleming2006} or \cite[Section 2]{crandall1992}.\\
Since $w$ is a viscosity subsolution of \eqref{eq:HJB}, the statement
$(D_{x,\theta}F^\alpha(z_\alpha),X)\in \bar{J}^{2,+}_{\bar{\Omega}}w(x_\alpha,\vartheta_\alpha)$
is equivalent to the existence of a subsolution test function at $(x_\alpha,\vartheta_\alpha)$ with first derivative equal to $D_{x,\theta}F^\alpha(z_\alpha)$ and second
derivative equal to $X$.
At the same time the statement $(-D_{y,\rho}F^\alpha(z_\alpha),Y)\in \bar{J}^{2,-}_{\bar{\Omega}}v(y_\alpha,\rho_\alpha)$ is equivalent to the existence of
a supersolution test function, again with the first derivative given by $-D_{y,\rho}F^\alpha(z_\alpha)$ and second derivative equal to $Y$.\\
Therefore, from the subsolution property of $w$ and the supersolution property of $v$ (using the above-mentioned derivatives for the respective test functions)  we derive
\begin{align*}
-\delta w(x_\alpha,\vartheta_\alpha)+\vartheta_\alpha\frac{(x_\alpha-y_\alpha)}{\alpha}
+\frac{1}{2}\trace[a(x_\alpha,\vartheta_\alpha)a(x_\alpha,\vartheta_\alpha)^TX]+\sup_{u\in[0,K]}u(1-\frac{(x_\alpha-y_\alpha)}{\alpha})\geq 0\,,\\
-\delta v(y_\alpha,\rho_\alpha)+\rho_\alpha\frac{(x_\alpha-y_\alpha)}{\alpha}
+\frac{1}{2}\trace[a(y_\alpha,\rho_\alpha)a(y_\alpha,\rho_\alpha)^TY]+\sup_{u\in[0,K]}u(1-\frac{(x_\alpha-y_\alpha)}{\alpha})\leq 0.
\end{align*}
Rearranging and using \eqref{eq:help3} yields
\begin{align*}
\delta(w(x_\alpha,\vartheta_\alpha)-v(y_\alpha,\rho_\alpha))
&\leq \frac{(x_\alpha-y_\alpha)(\vartheta_\alpha-\rho_\alpha)}{\alpha}+\frac{1}{2}\trace\left[\Sigma\left(\begin{array}{c c}
X & 0\\
0 &-Y
\end{array}
\right)\right]\\
&\leq \frac{(x_\alpha-y_\alpha)(\vartheta_\alpha-\rho_\alpha)}{\alpha}+\frac{3}{2\alpha}
(\vartheta_\alpha-\rho_\alpha)^2\frac{(\vartheta_\alpha-\theta_1-\theta_2+\rho_\alpha)^2}{\sigma^2}\\
&\leq \frac{\max\{(x_\alpha-y_\alpha)^2,(\vartheta_\alpha-\rho_\alpha)^2\}}{\alpha}+\frac{3}{2\alpha}
(\vartheta_\alpha-\rho_\alpha)^2\frac{(\vartheta_\alpha-\theta_1-\theta_2+\rho_\alpha)^2}{\sigma^2}.
\end{align*}
In the above inequality the left-hand side converges to $\delta M$ as $\alpha\to0$.
Since $F^\alpha(z_\alpha)\to0$ at the same time (other terms are bounded), the right-hand side converges to $0$,
resulting in the contradiction $\delta M\leq 0$, which concludes the proof of the theorem.
\end{proof}

\begin{proof}[Proof of Theorem \ref{thm:verification}]
Define $\varphi(x,\vartheta) := \frac{1}{\pi} e^{-(x^2+\vartheta^2)}$ and let
\[
 \varphi^n(x,\vartheta) := n^2 \int_{-\infty}^\infty \int_{- \infty}^\infty v(x-s,\vartheta-t) \varphi(ns,nt)\, ds\, dt\,.
\]
Note that as $n \to \infty$,
$\varphi^n \to v$ and $\cL \varphi^n \to \cL v$, see \citet{wheeden1977}.\\

Let $u=(u_t)_{t \ge 0}$ be an admissible strategy and let $T\in\N$. Then
 \begin{align*}
\lefteqn{e^{-\delta (T \wedge \tau)} \varphi^n (X_{T \wedge \tau},\theta_{T \wedge \tau})}\\
=& \varphi^n(x,\vartheta) + \int_0^{T \wedge \tau} e^{-\delta t}\, d\varphi^n(X_t,\theta_t) + \int_0^{T \wedge \tau} \varphi^n(X_t,\theta_t)\, d(e^{-\delta t})\\
=& \varphi^n(x,\vartheta) + \int_0^{T \wedge \tau}  e^{-\delta t} \left[ -\delta \varphi^n(X_t,\theta_t) + \cL \varphi^n (X_t, \theta_t) - u_t \varphi_x^n (X_t, \theta_t) \right]\,dt\\
+& \int_0^{T \wedge \tau} e^{-\delta t} \left[ \sigma  \varphi_x^n + \frac{(\theta_t-\theta_1)(\theta_2-\theta_t)}{\sigma} \varphi_\vartheta^n \right]\, dW_t\,.
\end{align*}

After taking expectations, the stochastic integral vanishes. Therefore,
\begin{align*}
& \E_{x,\vartheta} \left( e^{-\delta (T \wedge \tau)} \varphi^n (X_{T \wedge \tau},\vartheta_{T \wedge \tau})\right)\\
&= \varphi^n(x,\vartheta) + \E_{x,\vartheta} \left( \int_0^{T \wedge \tau}  e^{-\delta t} \left[ -\delta \varphi^n(X_t,\theta_t) + \cL \varphi^n (X_t, \theta_t) - u_t \varphi_x^n (X_t, \theta_t) \right]\,dt\right)\,.
\end{align*}

Let $\varepsilon>0$. Since $v$ fulfills
\[
 -\delta v + \cL v + (1-v_x) u \le 0 \quad \mbox{a.e.}
\]
we can choose  $n$ large enough such that 
\[
 -\delta \varphi^n + \cL \varphi^n + (1-\varphi^n_x) u \le \varepsilon\,,
\]
and hence
\[
 \cL \varphi^n \le \delta \varphi^n - (1-\varphi^n_x) u + \varepsilon\,.
\]
Therefore we get
\begin{align*}
&\E_{x,\vartheta} \left( e^{-\delta (T \wedge \tau)} \varphi^n (X_{T \wedge \tau},\theta_{T \wedge \tau})\right)\\
&\le \varphi^n(x,\vartheta) + \E_{x,\vartheta} \left( \int_0^{T \wedge \tau}  e^{-\delta t} \left[ -\delta \varphi^n(X_t,\theta_t) + \delta \varphi^n(X_t,\theta_t) - (1-\varphi_x^n(X_t,\theta_t)) u_t + \varepsilon - u_t \varphi_x^n (X_t, \theta_t) \right]\,dt\right)\\
&= \varphi^n(x,\vartheta) - \E_{x,\vartheta} \left( \int_0^{T \wedge \tau}  e^{-\delta t} u_t \,dt + \varepsilon \int_0^{T \wedge \tau}  e^{-\delta t} \,dt\right)\,.
\end{align*}

Letting $n\rightarrow\infty$, we get by dominated convergence
\begin{align*}
\E_{x,\vartheta} \left( e^{-\delta (T \wedge \tau)} v (X_{T \wedge \tau},\theta_{T \wedge \tau})\right)
\le v(x,\vartheta) - \E_{x,\vartheta} \left( \int_0^{T \wedge \tau}  e^{-\delta t} u_t \,dt + \varepsilon \int_0^{T \wedge \tau}  e^{-\delta t} \,dt\right)\,.
\end{align*}
As $\varepsilon$ was arbitrary, we further get
\[
\E_{x,\vartheta} \left( e^{-\delta (T \wedge \tau)} v (X_{T \wedge \tau},\theta_{T \wedge \tau})\right) \le v(x,\vartheta) - \E_{x,\vartheta} \left( \int_0^{T \wedge \tau}  e^{-\delta t} u_t \,dt\right)\,,
\]
and hence
\[
\E_{x,\vartheta} \left( e^{-\delta (T \wedge \tau)} v (X_{T \wedge \tau},\theta_{T \wedge \tau})\right) + \E_{x,\vartheta} \left( \int_0^{T \wedge \tau}  e^{-\delta t} u_t \,dt\right) \le v(x,\vartheta)\,.
\]

From the supersolution property we have that $\E_{x,\vartheta} \left(
e^{-\delta (T \wedge \tau)} v (X_{T \wedge \tau},\theta_{T \wedge \tau})\right)
\ge 0$. Thus we have, by Fatou's Lemma, 
\[
  J^{(u)}(x,\vartheta)=\E_{x,\vartheta} \left( \int_0^{\tau}  e^{-\delta t} u_t \,dt\right) 
\le \liminf_{T\rightarrow \infty} \E_{x,\vartheta} \left( \int_0^{\tau\wedge T}  e^{-\delta t} u_t \,dt\right) 
\le v(x,\vartheta)\,.
\]
So for each control $u$, $v$ dominates the value function. Taking the supremum over $u \in [0,K]$ in the derivation, we get equality in the HJB equation and therefore
\[
 V(x,\vartheta) \le v(x,\vartheta)\,,
\]
which completes the proof.
\end{proof}



\end{document}